%% file: arxiv_paper.tex
\documentclass[orivec,envcountsect,envcountsame]{llncs}

\usepackage{todonotes}

\usepackage{enumitem}
\usepackage{amssymb,amsmath}
\usepackage{url}
\usepackage{mathpartir}
\usepackage{stmaryrd}

\usepackage{comment}
\excludecomment{cameraversion}
\includecomment{arxivversion}

\usepackage{tikz-cd}

\makeatletter
\def\verbatim@font{\sffamily}
\makeatother

\begin{document}
\frontmatter          
\pagestyle{headings}  
\mainmatter
\title{Programming and Reasoning with \\ Guarded Recursion for Coinductive Types}

%
\titlerunning{Guarded Recursion for Coinductive Types}  
%
\author{Ranald Clouston \and Ale\v{s} Bizjak \and Hans Bugge Grathwohl‎ \and Lars Birkedal}
\authorrunning{Ranald Clouston et al.} 
\institute{Department of Computer Science, Aarhus University, Denmark\\
  \email{\{ranald.clouston,abizjak,hbugge,birkedal\}@cs.au.dk}}

\maketitle              

\input{macros}

%
\begin{abstract}
  We present the guarded lambda-calculus, an extension of the simply typed
  lambda-calculus with guarded recursive and coinductive types. The use of guarded
  recursive types ensures the productivity of well-typed programs. Guarded recursive
  types may be transformed into coinductive types by a type-former inspired by modal
  logic and Atkey-McBride clock quantification, allowing the typing of acausal functions.
  We give a call-by-name operational semantics for the calculus, and define adequate
  denotational semantics in the topos of trees. The adequacy proof entails that the
  evaluation of a program always terminates. We demonstrate the expressiveness of
  the calculus by showing the definability of solutions to Rutten's behavioural
  differential equations. We introduce a program logic with L{\"o}b induction for
  reasoning about the contextual equivalence of programs.
\end{abstract}
%

\input{intro}

\input{language}

\input{logic}

\input{bde-examples}

\input{conclusion}

\subsection*{Acknowledgements}
We gratefully acknowledge our discussions with Andreas Abel, Tadeusz Litak,
Stefan Milius, Rasmus M{\o}gelberg, Filip Sieczkowski, and Andrea Vezzosi,
and the comments of the reviewers.
This research was supported in part by the \mbox{ModuRes} Sapere Aude Advanced Grant from The Danish Council for Independent Research for the Natural Sciences (FNU). Ale\v{s} Bizjak is supported in part by a Microsoft Research PhD grant.

\bibliographystyle{splncs03}
\bibliography{paper}

\clearpage
\appendix

\input{appx_lang_proofs}

\input{appx_example_proofs}

\input{appx_sums}

\input{appx_bde_proof}

\input{appx_total_inhab}

\end{document}

%% file: macros.tex
\renewcommand{\labelenumi}{(\roman{enumi})}

\newcommand{\lambdanext}{\ensuremath{\mathsf{g}\lambda}}
\newcommand{\logiclambdanext}{\ensuremath{L\mathsf{g}\lambda}}
\newcommand{\NEXT}{\operatorname{\mathsf{next}}}
\newcommand{\LATER}{{\blacktriangleright}}
\newcommand{\UNFOLD}{\operatorname{\mathsf{unfold}}}
\newcommand{\FOLD}{\operatorname{\mathsf{fold}}}
\newcommand{\ABORT}{\operatorname{\mathsf{abort}}}
\newcommand{\UNIT}{\operatorname{\langle\rangle}}
\newcommand{\THEN}{\operatorname{\mathsf{then}}}
\newcommand{\ELSE}{\operatorname{\mathsf{else}}}
\newcommand{\CASE}{\operatorname{\mathsf{case}}}
\newcommand{\OF}{\operatorname{\mathsf{of}}}
\newcommand{\IN}{\operatorname{\mathsf{in}}}
\newcommand{\EMPTY}{\operatorname{\mathbf{0}}}
\newcommand{\ONE}{\operatorname{\mathbf{1}}}
\newcommand{\APP}{\circledast}
\newcommand{\NAT}{\operatorname{\mathbf{N}}}
\newcommand{\ZERO}{\operatorname{\mathsf{zero}}}
\newcommand{\SUCC}{\operatorname{\mathsf{succ}}}
\newcommand{\BOX}{\operatorname{\mathsf{box}}}
\newcommand{\WITH}{\operatorname{\mathsf{with}}}
\newcommand{\PREV}{\operatorname{\mathsf{prev}}}
\newcommand{\UNBOX}{\operatorname{\mathsf{unbox}}}
\newcommand{\BOXSUM}{\operatorname{\mathsf{box}^+}}

\newcommand{\Exp}{\operatorname{Exp}}
\newcommand{\ClType}{\operatorname{ClType}}
\newcommand{\red}{\mathrel{\mapsto}}
\newcommand{\redrt}{\mathrel{\rightsquigarrow}}
\newcommand{\llrr}[1]{\llbracket #1 \rrbracket}
\newcommand{\defeq}{\triangleq}
\newcommand{\bnfeq}{\mathrel{::=}}
\newcommand{\trees}{\mathcal{S}}
\newcommand{\res}[2]{r^{#1}_{#2}}
\newcommand{\den}[1]{\llbracket#1\rrbracket}
\newcommand{\lrel}[2]{R^{#2}_{#1}}
\newcommand{\boxd}{\mathsf{bd}}
\newcommand{\usize}{\mathsf{us}}
\newcommand{\natto}{\mathrel{\dot{\to}}}
\newcommand{\ceq}{\simeq_{\mathsf{ctx}}}

\newcommand{\ThetaLetter}{\mathchar"7002}
\renewcommand{\Theta}{\operatorname{\mathsf{fix}}}

\newcommand{\explsubst}{\leftarrow}

\newcommand{\tinylater}{\scriptscriptstyle\LATER}

\newcommand{\guarded}[1]{\ensuremath{#1^{\mathsf{g}}}}

\newcommand{\gStream}[1]{\guarded{\mathsf{Str}}}
\newcommand{\Stream}[1]{\mathsf{Str}}
\newcommand{\head}{\operatorname{\guarded{\mathsf{hd}}}}
\newcommand{\tail}{\operatorname{\guarded{\mathsf{tl}}}}
\newcommand{\limhead}{\operatorname{\mathsf{hd}}}
\newcommand{\limtail}{\operatorname{\mathsf{tl}}}

\newcommand{\iterate}{\operatorname{\mathsf{iterate}}}
\newcommand{\cons}{\operatorname{\mathsf{cons}}}
\newcommand{\consin}{\mathbin{::}}
\newcommand{\gsecond}{\guarded{\mathsf{2nd}}}
\newcommand{\gthird}{\guarded{\mathsf{3rd}}}
\newcommand{\decons}{\operatorname{\mathsf{decons}}}
\newcommand{\naturals}{\operatorname{\mathsf{nats}}}
\renewcommand{\interleave}{\operatorname{\mathsf{interleave}}}
\newcommand{\tyrol}{\operatorname{\mathsf{toggle}}}
\newcommand{\folds}{\operatorname{\mathsf{paperfolds}}}
\newcommand{\limit}{\operatorname{\mathsf{lim}}}
\renewcommand{\lim}{\limit}
\newcommand{\everysecond}{\operatorname{\mathsf{every2nd}}}
\newcommand{\plus}{\operatorname{\guarded{\mathsf{plus}}}}
\newcommand{\map}{\operatorname{\guarded{\mathsf{map}}}}
\newcommand{\limplus}{\operatorname{\mathsf{plus}}}

\newcommand{\hastype}[3]{\ensuremath{{#1 \vdash #2 : #3}}}

\newcommand{\eps}{\varepsilon}

\newcommand{\El}{\ensuremath{\mathcal{E}}}
\newcommand{\Sl}{\ensuremath{\mathcal{S}}}
\newcommand{\Ul}{\ensuremath{\mathcal{U}}}
\newcommand{\Dl}{\ensuremath{\mathcal{D}}}
\newcommand{\Fl}{\ensuremath{\mathcal{F}}}
\newcommand{\Pl}{\ensuremath{\mathcal{P}}}
\newcommand{\Tl}{\ensuremath{\mathcal{T}}}
\newcommand{\Ml}{\ensuremath{\mathcal{M}}}
\newcommand{\Il}{\ensuremath{\mathcal{I}}}
\newcommand{\Cl}{\ensuremath{\mathcal{C}}}
\newcommand{\Bl}{\ensuremath{\mathcal{B}}}
\newcommand{\Al}{\ensuremath{\mathcal{A}}}
\newcommand{\Gl}{\ensuremath{\mathcal{G}}}
\newcommand{\Nl}{\ensuremath{\mathcal{N}}}
\newcommand{\BB}{\ensuremath{\mathbb{B}}}
\newcommand{\CC}{\ensuremath{\mathbb{C}}}
\newcommand{\KK}{\ensuremath{\mathbb{K}}}
\newcommand{\NN}{\NAT}
\newcommand{\PP}{\ensuremath{\mathbb{P}}}
\newcommand{\VV}{\ensuremath{\mathbb{V}}}
\newcommand{\UU}{\ensuremath{\mathbb{U}}}
\newcommand{\DD}{\ensuremath{\mathbb{D}}}
\newcommand{\EE}{\ensuremath{\mathbb{E}}}
\newcommand{\TT}{\ensuremath{\mathbb{T}}}

\newcommand{\comp}{\circ}

\newcommand{\id}[1]{\ensuremath{\text{id}_{#1}}}
\newcommand{\inv}[1]{\ensuremath{#1^{-1}}}
\newcommand{\iso}{\cong}
\newcommand{\isetsep}{\ensuremath{{\,\middle|\,}}}

\newcommand{\later}{\operatorname\triangleright}
\newcommand{\always}{\operatorname\square}
\newcommand{\lift}{\operatorname{\mathsf{lift}}}
\newcommand{\liftstr}[1]{#1_{\guarded{\mathsf{Str}}}}

\newcommand{\eqlaternextrule}{\ensuremath{\textsc{eq}^{\later}_{\NEXT}}}

\renewcommand{\implies}{\Rightarrow}
\renewcommand{\iff}{\Leftrightarrow}

\newcommand{\inhab}[1]{\ensuremath{\mathrm{Inhab}\left(#1\right)}}
\newcommand{\total}[1]{\ensuremath{\mathrm{Total}\left(#1\right)}}

\newcommand{\defined}{\ensuremath{\overset{\triangle}{=}}}

\newcommand{\op}[1]{\ensuremath{#1^{\text{op}}}}
\newcommand{\CAT}{\ensuremath{\mathbf{Cat}}}
\newcommand{\sets}{\ensuremath{\mathbf{Set}}}
\newcommand{\Sh}[1]{\ensuremath{\text{Sh}\left(#1\right)}}
\newcommand{\PSh}[1]{\ensuremath{\text{PSh}\left(#1\right)}}
\newcommand{\subobj}[1]{\ensuremath{\mathbf{Sub}\left(#1\right)}}
\renewcommand{\hom}[3]{\ensuremath{\text{Hom}_{#1}\left(#2,#3\right)}}
\newcommand{\nxt}{\ensuremath{\mathbf{next}}}

\newenvironment{diagram}{\begin{tikzcd}[row sep=1.5cm,column sep=1.5cm]}{\end{tikzcd}}
\newenvironment{largediagram}{\begin{tikzcd}[row sep=2.6cm,column sep=2.6cm]}{\end{tikzcd}}
\newenvironment{smalldiagram}{\begin{tikzcd}[row sep=1cm,column sep=1cm]}{\end{tikzcd}}

\newcommand{\denS}[1]{\ensuremath{\left\llbracket #1 \right\rrbracket_{\Sl}}}
\newcommand{\denSet}[1]{\ensuremath{\left\llbracket #1 \right\rrbracket_{\sets}}}


%% file: intro.tex
\section{Introduction}
\label{sec:introduction}

The problem of ensuring that functions on coinductive types are well-defined has
prompted a wide variety of work into productivity checking, and rule formats for
coalgebra.
\emph{Guarded recursion}~\cite{Coquand:Infinite} guarantees productivity
and unique solutions by requiring
that recursive calls be nested under a constructor, such as cons (written $\consin$) for
streams. This can sometimes be established by a simple syntactic check, as for
the stream $\mathsf{toggle}$ and binary stream function $\interleave$ below:
\begin{verbatim}
  toggle = 1 :: 0 :: toggle
  interleave (x :: xs) ys = x :: interleave ys xs
\end{verbatim}
Such syntactic checks, however, are often too blunt and exclude many
valid definitions. For example the \emph{regular paperfolding sequence}, the sequence of left and right
turns (encoded as $1$ and $0$) generated by repeatedly folding a piece of paper in
half, can be defined via the function $\interleave$ as
follows~\cite{Endrullis:Mix}:
\begin{verbatim}
  paperfolds = interleave toggle paperfolds
\end{verbatim}
This definition is productive, but the putative definition below, which also applies
$\interleave$ to two streams and so apparently is just as well-typed, is not:
\begin{verbatim}
  paperfolds' = interleave paperfolds' toggle
\end{verbatim}
This equation is satisfied by any stream whose \emph{tail} is the
regular paperfolding sequence, so lacks a unique solution. Unfortunately the syntactic
productivity checker of the proof assistant Coq~\cite{Gimenez:Codifying} will reject
both definitions.

A more flexible approach, first suggested by Nakano~\cite{Nakano:Modality}, is to
guarantee productivity via \emph{types}. A new modality, for which we follow Appel et
al.~\cite{Appel:Very} by writing $\LATER$ and using the name `later', allows us to
distinguish between data we have access to now, and data which we have only later.
This $\LATER$ must be used to guard self-reference in type definitions, so for example
\emph{guarded streams} of natural numbers are defined by the guarded recursive
equation
\[
  \gStream{\NAT} \defeq \NAT\times \LATER\gStream{\NAT}
\]
asserting that stream heads are available now, but tails only later.
The type of $\interleave$ will be $\gStream{\NAT}\to\LATER\gStream{\NAT}\to
\gStream{\NAT}$, capturing the fact the (head of the) first argument is needed
immediately, but the second argument is needed only later. In term definitions the
types of self-references will then be guarded by $\LATER$ also.
For example $\interleave\folds'\mathsf{toggle}$ becomes ill-formed, as the
$\folds'$ self-reference has type $\LATER\gStream{\NAT}$, rather than
$\gStream{\NAT}$, but $\interleave\mathsf{toggle}\,\folds$ will be well-formed. 

Adding $\LATER$ alone to the simply typed $\lambda$-calculus enforces a discipline
more rigid than productivity. For example the obviously productive stream function
\begin{verbatim}
  every2nd (x :: x' :: xs) = x :: every2nd xs
\end{verbatim}
cannot be typed because it violates \emph{causality}~\cite{Krishnaswami:Ultrametric}:
elements of the result stream depend on deeper elements of the argument stream. In
some settings, such as reactive programming, this is a desirable property, but for
productivity guarantees alone it is too restrictive. We need the ability to remove
$\LATER$ in a controlled way. This is provided by the \emph{clock quantifiers}
of Atkey and McBride~\cite{Atkey:Productive}, which assert that all data is available
now. This does not trivialise the guardedness requirements because there
are side-conditions controlling when clock quantifiers may be introduced. Moreover
clock quantifiers transform guarded recursive types into first-class \emph{coinductive}
types, with guarded recursion defining the rule format for their manipulation.

Our presentation departs from Atkey and McBride's~\cite{Atkey:Productive} by
regarding the
`everything now' operator as a unary type-former, written $\blacksquare$ and called
`constant', rather than a quantifier. Observing that the types $\blacksquare A\to A$ and
$\blacksquare A\to\blacksquare\blacksquare A$ are always inhabited allows us to see
the type-former, via the Curry-Howard isomorphism, as an \emph{S4} modality, and
hence base our operational semantics on the established typed calculi for intuitionistic
S4 (IS4) of Bierman and de Paiva~\cite{Bierman:Intuitionistic}. This is sufficient to
capture all examples in the literature, which use only one clock; for examples that
require multiple clocks 
we suggest extending our calculus to a \emph{multimodal} logic.

\paragraph{In this paper} we present the guarded $\lambda$-calculus,
$\lambdanext$, extending the simply typed $\lambda$-calculus
with coinductive and guarded recursive types. We define call-by-name operational
semantics, which blocks non-termination via recursive definitions unfolding indefinitely.
We define adequate denotational semantics in the topos of
trees~\cite{Birkedal-et-al:topos-of-trees} and as a consequence prove normalisation.
We introduce a program logic $\logiclambdanext$ for reasoning about the
denotations of $\lambdanext$-programs; given adequacy this permits proofs about the
operational behaviour of terms.
The logic is based on the internal logic of the
topos of trees, with modalities $\later,\always$ on
predicates, and L{\"o}b induction for
reasoning about functions on both guarded recursive and coinductive types.
We demonstrate the
expressiveness of the calculus by showing the definability of
solutions to Rutten's behavioural differential
equations~\cite{Rutten:2003:bde}, and show that $\logiclambdanext$ can be used
to reason about them, as an alternative to standard bisimulation-based arguments.

We have implemented the $\lambdanext$-calculus in Agda, a process we found helpful
when fine-tuning the design of our calculus. The implementation, with many examples,
is available at \url{http://cs.au.dk/~hbugge/gl-agda.zip}.


%% file: language.tex

\section{Guarded $\lambda$-calculus}\label{sec:calculus}

This section presents the guarded $\lambda$-calculus, written $\lambdanext$, its
call-by-name operational semantics, and its types, then gives some examples.


\begin{definition}\label{def:terms}
  $\lambdanext$-\emph{terms} are given by the grammar
  \[
    \begin{array}{rcl}
      t & \bnfeq & x ~|~  \UNIT ~|~ \ZERO ~|~ \SUCC t ~|~ \langle t,t \rangle ~|~
      \pi_d t ~|~ \lambda x . t ~|~ tt ~|~ \FOLD t ~|~ \UNFOLD t  \\
    &|&   \NEXT t ~|~ \PREV \sigma.t ~|~
      \BOX \sigma.t ~|~\UNBOX t ~|~ t \APP t
    \end{array}
  \]
  where $d\in\{1,2\}$, $x$ is a variable and $\sigma = [x_1 \explsubst t_1,\ldots,
  x_n\explsubst t_n]$, usually abbreviated $[ \vec{x} \explsubst \vec{t} ]$,
  is a list of variables paired with terms.
 
  $\PREV [\vec{x}\explsubst\vec{t}].t$ and
  $\BOX [\vec{x}\explsubst\vec{t}].t$ bind all variables of $\vec{x}$ in $t$, but
  \emph{not} in $\vec{t}$. We write $\PREV \iota. t$ for
  $\PREV[\vec{x}\explsubst\vec{x}] .t$
  where $\vec{x}$ is a list of all free variables of $t$. If furthermore $t$ is closed we
  simply write $\PREV t$. We will similarly write $\BOX\iota.t$ and
  $\BOX t$. We adopt the convention that $\PREV$ and $\BOX$ have highest
  precedence.
\end{definition}

We may extend $\lambdanext$ with sums; for space
reasons
\begin{cameraversion}
these appear only in the extended version of this paper~\cite{ARXIVVERSION}.
\end{cameraversion}
\begin{arxivversion}
we leave these to App.~\ref{app:sums}.
\end{arxivversion}
\begin{definition}\label{def:redrule}
The \emph{reduction rules} on closed $\lambdanext$-terms are
\[
  \begin{array}{rcll}
    \pi_d \langle t_1, t_2 \rangle & \red & t_d & \quad\mbox{\emph{($d\in\{1,2\}$)}} \\
    (\lambda x . t_1) t_2 & \red & t_1 [t_2/x] \\
    \UNFOLD \FOLD t & \red & t \\
    \PREV [\vec{x} \explsubst \vec{t}].t & \red & \PREV t[\vec{t}/\vec{x}]
      & \quad\mbox{\emph{($\vec{x}$ non-empty)}} \\
    \PREV\NEXT t & \red & t \\
    \UNBOX(\BOX [\vec{x} \explsubst \vec{t}].t) & \red & t[\vec{t}/\vec{x}] \\
    \NEXT t_1\APP\NEXT t_2 & \red & \NEXT(t_1 t_2)
  \end{array} 
\]
\end{definition}

The rules above look like standard $\beta$-reduction, removing `roundabouts' of
introduction then elimination, with the exception of those regarding $\PREV$ and
$\NEXT$. An apparently more conventional $\beta$-rule for these term-formers would be
\[
  \PREV [\vec{x} \explsubst \vec{t}].(\NEXT t) \; \red \; t[\vec{t}/\vec{x}]
\]
but where $\vec{x}$ is non-empty this would require us to reduce an open term
to derive $\NEXT t$. We take the view that reduction of open terms is undesirable
within a call-by-name discipline, so first apply the substitution without eliminating
$\PREV$.

The final rule is not a true $\beta$-rule, as $\APP$ is neither
introduction nor elimination, but is necessary to enable function application under a
$\NEXT$ and hence allow, for example, manipulation of the tail of a stream. It
corresponds to the `homomorphism' equality for applicative
functors~\cite{McBride:Applicative}.

We next impose our call-by-name strategy on these reductions.
\begin{definition}\label{def:value}
\emph{Values} are terms of the form
\[
  \UNIT~|~ \SUCC^n\ZERO ~|~  \langle t,t\rangle ~|~
  \lambda x.t ~|~ \FOLD t ~|~ \BOX \sigma.t  ~|~ \NEXT t
\]
where $\SUCC^n$ is a list of zero or more $\SUCC$ operators, and $t$ is any term.
\end{definition}
\begin{definition}\label{def:eval_ctx}
  \emph{Evaluation contexts} are defined by the grammar
  \[
    \begin{array}{rcl}
      E &\bnfeq& \cdot ~|~
      \SUCC E ~|~ \pi_d E ~|~ E t ~|~ \UNFOLD E ~|~
      \PREV E ~|~ \UNBOX E ~|~ E \APP t ~|~ v \APP E
    \end{array} 
  \]
\end{definition}

If we regard $\APP$ as a variant of function application, it is surprising in a
call-by-name setting to reduce on both its sides. However both sides must be reduced
until they have main connective $\NEXT$ before the reduction rule for $\APP$ may be
applied. Thus the order of reductions of $\lambdanext$-terms cannot be identified with
the call-by-name reductions of the corresponding $\lambda$-calculus term with the
novel connectives erased.
\begin{definition}
  \emph{Call-by-name reduction} has format $E[t]\red E[u]$, where $t\red u$ is a
  reduction rule. From now the symbol $\red$ will be reserved to refer to
  call-by-name reduction. We use $\redrt$ for the reflexive transitive closure of $\red$.
\end{definition}
\begin{lemma}\label{lem:det}
The call-by-name reduction relation $\red$ is deterministic.
\end{lemma}


\begin{definition}\label{def:types}
  $\lambdanext$-\emph{types} are defined inductively by the rules of
  Fig.~\ref{fig:types}. $\nabla$ is a finite set of \emph{type variables}. A variable
  $\alpha$ is \emph{guarded in} a type $A$ if all
  occurrences of $\alpha$ are beneath an occurrence of $\LATER$ in the syntax tree.
 We adopt the convention that unary type-formers bind closer than binary type-formers.
\end{definition}

\begin{figure}
  \begin{mathpar}
    \inferrule*{ }{\nabla,\alpha \vdash \alpha}
    \and
    \inferrule*{ }{\nabla \vdash \ONE}
    \and
    \inferrule*{ }{\nabla \vdash \NAT}
    \and
    \inferrule*{%
      \nabla \vdash A_1 \\
      \nabla \vdash A_2}{%
      \nabla\vdash A_1\times A_2}
    \and
    \inferrule*{%
      \nabla \vdash A_1 \\
      \nabla \vdash A_2}{%
      \nabla\vdash A_1\to A_2}
    \and
    \inferrule*[right={$\alpha\,\mathsf{guarded\,in}\,A$}]{%
      \nabla,\alpha \vdash A}{%
      \nabla \vdash \mu\alpha.A}
    \and
    \inferrule*{%
      \nabla \vdash A}{%
      \nabla\vdash \LATER A}
    \and
    \inferrule*{%
      \cdot \vdash A}{%
      \nabla\vdash \blacksquare A}
  \end{mathpar}
  \caption{Type formation for the $\lambdanext$-calculus}
  \label{fig:types}
\end{figure}

Note the side condition on the $\mu$ type-former, and the prohibition on $\blacksquare
A$ for open $A$, which can also be understood as a prohibition on applying
$\mu\alpha$ to any $\alpha$ with $\blacksquare$ above it. The intuition for these
restrictions is that unique fixed points exist only where the variable is displaced
in time by a $\LATER$, but $\blacksquare$ cancels out this displacement by giving
`everything now'.
\begin{definition}\label{def:typing}
The \emph{typing judgments} are given in Fig.~\ref{fig:typing}. There $d\in\{1,2\}$,
and the \emph{typing contexts} $\Gamma$
are finite sets of pairs $x:A$ where $x$ is a variable and $A$ a closed type.
Closed types are \emph{constant} if all occurrences of $\LATER$ are beneath an
occurrence of $\blacksquare$ in their syntax tree.
\end{definition}

\begin{figure}
  \begin{mathpar}
    \inferrule*{ }{\Gamma, x : A \vdash x : A}
    \and
    \inferrule*{ }{\Gamma \vdash \UNIT : \ONE}
    \and
    \inferrule*{ }{\Gamma \vdash \ZERO:\NAT}
    \and
    \inferrule*{\Gamma \vdash t:\NAT}{%
      \Gamma \vdash \SUCC t:\NAT}
    \and
    \inferrule*{\Gamma \vdash t_1 : A \\%
      \Gamma \vdash t_2 : B}{%
      \Gamma \vdash \langle t_1,t_2 \rangle : A \times B}
    \and
    \inferrule*{\Gamma \vdash t: A_1 \times A_2}{%
      \Gamma \vdash \pi_d t : A_d}
    \and
    \inferrule*{\Gamma, x : A \vdash t : B}{%
      \Gamma \vdash \lambda x . t : A \to B}
    \and
    \inferrule*{\Gamma \vdash t_1 : A \to B \\%
      \Gamma \vdash t_2 : A}{%
      \Gamma \vdash t_1 t_2 : B}
    \and
    \inferrule*{\Gamma \vdash t:A[\mu\alpha.A/\alpha]}{%
      \Gamma \vdash \FOLD t : \mu\alpha.A}
    \and
    \inferrule*{\Gamma \vdash t : \mu\alpha.A}{%
      \Gamma \vdash \UNFOLD t : A[\mu\alpha.A/\alpha]}
   \and
   \inferrule*{\Gamma \vdash t : A}{%
     \Gamma \vdash \NEXT t : \LATER A}
    \and
     \inferrule*[right={$A_1,\ldots,A_n\,\mathsf{constant}$}]{%
      x_1:A_1,\ldots,x_n:A_n \vdash t:\LATER A \\
      \Gamma\vdash t_1:A_1 \\
      \cdots \\
      \Gamma\vdash t_n:A_n }{%
      \Gamma \vdash \PREV [x_1\explsubst t_1,\ldots,x_n\explsubst t_n].t : A}
    \and
    \inferrule*[right={$A_1,\ldots,A_n\,\mathsf{constant}$}]{%
      x_1:A_1,\ldots,x_n:A_n \vdash t:A \\
      \Gamma\vdash t_1:A_1 \\
      \cdots \\
      \Gamma\vdash t_n:A_n }{%
      \Gamma \vdash \BOX [x_1\explsubst t_1,\ldots,x_n\explsubst t_n].t :\blacksquare A}
    \and
    \inferrule*{\Gamma \vdash t:\blacksquare A}{%
      \Gamma \vdash \UNBOX t: A}
    \and
    \inferrule*{\Gamma \vdash t_1 : \LATER (A \to B) \\%
      \Gamma \vdash t_2 : \LATER A}{%
      \Gamma \vdash t_1 \APP t_2 : \LATER B}
  \end{mathpar}
  \caption{Typing rules for the $\lambdanext$-calculus}
  \label{fig:typing}
\end{figure}

The \emph{constant} types exist `all at once', due to the absence of $\LATER$
or presence of $\blacksquare$; this condition corresponds to the freeness of the
clock variable in Atkey and McBride~\cite{Atkey:Productive} (recalling that we use only
one clock in this
work). Its use as a side-condition to $\blacksquare$-introduction in Fig.~\ref{fig:typing}
recalls (but is more
general than) the `essentially modal' condition for natural deduction for IS4 of
Prawitz~\cite{Prawitz:Natural}. The term calculus for IS4 of Bierman and de
Paiva~\cite{Bierman:Intuitionistic},
on which this calculus is most closely based, uses the still more restrictive requirement
that $\blacksquare$ be the main connective. 
This would preclude some functions that
seem desirable, such as the isomorphism $\lambda n.\BOX \iota.n:\NAT\to\blacksquare
\NAT$.

In examples $\PREV$ usually appears in its syntactic sugar forms
\[
  \inferrule*[right={$A_1,\ldots,A_n\,\mathsf{constant}$}]{%
    x_1:A_1,\ldots,x_n:A_n \vdash t:\LATER A }{%
    \Gamma,x_1:A_1,\ldots,x_n:A_n \vdash \PREV \iota. t: A}
  \qquad
  \inferrule*{%
    \vdash t : \LATER A}{%
    \Gamma \vdash \PREV t : A}
\]
and similarly for $\BOX$; the more general form is nonetheless necessary because
$(\PREV \iota. t)[\vec{u}/\vec{x}] = \PREV
[\vec{x}\explsubst\vec{u}].t$. Getting
substitution right in this setting is somewhat delicate. For example our reduction
rule $\PREV [\vec{x}\explsubst \vec{t}].t \red\PREV t[\vec{t}/\vec{x}]$ breaches
subject reduction on open terms (but not for closed terms). See Bierman and de
Paiva~\cite{Bierman:Intuitionistic} for more discussion of substitution with respect
to IS4.
\begin{lemma}[Subject Reduction]
$\vdash t:A$ and $t\redrt u$ implies $\vdash u:A$.
\end{lemma}
\begin{example}\label{ex:programs}
\begin{enumerate}
\item
  The type of guarded recursive streams of natural numbers, $\gStream{\NAT}$, is
  defined as $\mu\alpha.\NAT\times\LATER\alpha$. These provide the setting for all
  examples below, but other definable types include infinite binary trees, as $\mu\alpha.
  \NAT\times\LATER\alpha\times\LATER\alpha$, and potentially infinite lists, as
  $\mu\alpha.\ONE+(\NAT\times\LATER\alpha)$.
\item
  We define guarded versions of the standard stream functions cons (written infix as
  $\consin$), head, and tail as obvious:
  \[
    \begin{array}{c}
    \consin \defeq \lambda n.\lambda s.\FOLD\langle n,s\rangle:
      \NAT\to\LATER\gStream{\NAT}\to\gStream{\NAT} \\
    \head \defeq \lambda s.\pi_1\UNFOLD s:\gStream{\NAT}\to\NAT \quad
    \tail \defeq \lambda s.\pi_2\UNFOLD s::\gStream{\NAT}\to\LATER\gStream{\NAT}
    \end{array}
  \]
  then use the $\APP$ term-former for observations deeper into the stream:
  \[
    \begin{array}{rcl}
    \gsecond &\defeq&
      \lambda s.(\NEXT\head)\APP(\tail s):\gStream{\NAT}\to\LATER\NAT \\
    \gthird &\defeq&
      \lambda s.(\NEXT\gsecond)\APP(\tail s):\gStream{\NAT}\to\LATER\LATER\NAT
      \;\cdots
    \end{array}
  \]
\item
  Following Abel and Vezzosi~\cite[Sec. 3.4]{Abel:Formalized} we may define a fixed
  point combinator $\Theta$ with type $(\LATER A \to A) \to A$ for any $A$.
  We use this to define a stream by iteration of a function: $\iterate$
  takes as arguments a natural number and a function, but the function is not used until
  the `next' step of computation, so we may reflect this with our typing:
  \[
    \iterate \defeq \lambda f . \Theta \lambda g . \lambda n . n \consin (g
    \APP (f \APP \NEXT n) ) : \LATER (\NAT \to \NAT) \to \NAT \to \gStream{\NAT}
  \]
  We may hence define the guarded stream of natural numbers
  \[
    \naturals \defeq \iterate \, (\NEXT\lambda n.\SUCC n) \ZERO.
  \]
\item
  With $\interleave$, following our discussion in the introduction, we again may reflect
  in our type that one of our arguments is not required until the next step, defining the
  term $\interleave$ as:
  \[
    \Theta\lambda g.\lambda s.\lambda t.(\head s)\consin
      (g\APP t\APP \NEXT(\tail s)) :
      \gStream{\NAT}\to\LATER\gStream{\NAT}\to\gStream{\NAT}
  \]
  This typing decision is essential to define the paper folding stream:
  \[
    \begin{array}{rcl}
      \tyrol &\defeq& \Theta\lambda s.(\SUCC\ZERO)\consin
        (\NEXT(\ZERO\consin s)) \\
      \folds &\defeq& \Theta\lambda s.\interleave\tyrol\,s
     \end{array}
  \]
  Note that the unproductive definition with $\interleave s\,\tyrol$ cannot be made to
  type check: informally, $s:\LATER\gStream{\NAT}$ cannot be converted into a
  $\gStream{\NAT}$ by $\PREV$, as it is in the scope of a
  variable $s$ whose type $\gStream{\NAT}$ is not constant. To see a less articifial
  non-example, try to define a $\mathsf{filter}$ function on streams which eliminates
  elements that fail some boolean test.
\item
  $\mu$-types are in fact \emph{unique} fixed points, so carry both final
  coalgebra and initial algebra structure. To see the latter, observe that we can define
  \[
    \mathsf{foldr} \defeq \Theta\lambda g\lambda f.\lambda s.f
      \langle\head s,g\APP\NEXT f\APP \tail s\rangle:
      ((\NAT\times\LATER A)\to A)\to \gStream{\NAT}\to A
  \]
  and hence for example $\mathsf{\map}\,h:\gStream{\NAT}\to\gStream{\NAT}$ is
  $\mathsf{foldr}\,\lambda x.(h \pi_1 x)\consin(\pi_2 x)$.
\item
  The $\blacksquare$ type-former lifts guarded recursive streams to
  coinductive streams, as we will make precise in Ex.~\ref{ex:denote_streams}. Let
  $\Stream{\NAT} \defeq \blacksquare\gStream{\NAT}$. We define $\limhead :
  \Stream{\NAT} \to \NAT$ and $\limtail : \Stream{\NAT} \to \Stream{\NAT}$ by
  $\limhead = \lambda s . \head (\UNBOX s)$ and
  $\limtail = \lambda s . \BOX \iota . \PREV \iota . \tail (\UNBOX s)$, and
  hence define observations deep into streams whose
  results bear no trace of $\LATER$, for example
  $\mathsf{2nd} \defeq \lambda s.\limhead (\limtail s) : \Stream{\NAT} \to \NAT$.
  
  In general boxed functions lift to functions on boxed types by
  \[
    \limit \defeq \lambda f.\lambda x.\BOX \iota.(\UNBOX f)(\UNBOX x):
      \blacksquare(A\to B)\to\blacksquare A\to\blacksquare B
  \]
\item
  The more sophisticated acausal function $\everysecond:\Stream{\NAT}\to\gStream
  \NAT$ is
  \[
    \Theta\lambda g.\lambda s.(\limhead s)\consin (g\APP(\NEXT (\limtail (\limtail s)))).
  \]
  Note that it must take a \emph{coinductive} stream $\Stream{\NAT}$ as argument.  The
  function with coinductive result type is then
  $\lambda s.\BOX \iota.\everysecond s:\Stream{\NAT}\to\Stream{\NAT}$.
\end{enumerate}
\end{example}

\section{Denotational Semantics and Normalisation}\label{sec:denot}

This section gives denotational semantics for $\lambdanext$-types and terms, as
objects and arrows in the topos of trees~\cite{Birkedal-et-al:topos-of-trees}, the
presheaf category over the first infinite ordinal $\omega$ (we give a concrete definition
below). These semantics are shown to be sound and, by a logical relations argument,
adequate with respect to
the operational semantics. Normalisation follows as a corollary of
this argument. Note that for space reasons many proofs, and some lemmas,
appear
\begin{cameraversion}
only in the extended version of this paper~\cite{ARXIVVERSION}.
\end{cameraversion}
\begin{arxivversion}
in App.~\ref{app:lang_proofs}.
\end{arxivversion}
\begin{definition}
The \emph{topos of trees} $\trees$ has, as objects $X$, families of sets $X_1,X_2,$
$\ldots$ indexed by the positive integers, equipped with families of \emph{restriction
functions} $\res{X}{i}:X_{i+1}\to X_i$ indexed similarly. Arrows $f:X\to Y$ are families
of functions $f_i:X_i\to Y_i$ indexed similarly obeying the naturality condition $f_i\circ
\res{X}{i}=\res{Y}{i}\circ f_{i+1}$.
\end{definition}

$\trees$ is a cartesian closed category with products defined pointwise. Its exponential
$A^B$ has, as its component sets $(A^B)_i$, the set of $i$-tuples $(f_1:A_1\to
B_1,\ldots,f_i:A_i\to B_i)$ obeying the naturality condition, and projections as
restriction functions.
\begin{definition}\label{def:functors}
\begin{itemize}
\item
  The category of sets $\sets$ is a full subcategory of $\trees$ via the functor $\Delta:
  \sets\to\trees$ with $(\Delta Z)_i=Z$, $\res{\Delta Z}{i}=id_Z$, and $(\Delta f)_i=f$.
  Objects in this subcategory are called \emph{constant objects}. In particular the
  terminal object $1$ of $\trees$ is $\Delta \{\ast\}$ and the \emph{natural numbers
  object} is $\Delta\mathbb{N}$;
\item
  $\Delta$ is left adjoint to $hom_{\trees}(1,\mbox{--})$; write $\blacksquare$ for
  $\Delta\circ hom_{\trees}(1,\mbox{-}):\trees\to\trees$. $\UNBOX:\blacksquare
  \natto id_{\trees}$ is the counit of the resulting comonad. Concretely $\UNBOX_i(x)
  =x_i$, i.e. the $i$'th component of $x:1\to X$ applied to $\ast$;
\item
  $\LATER:\trees\to\trees$ is defined by $(\LATER X)_1=\{\ast\}$ and $(\LATER
  X)_{i+1}=X_i$, with $\res{\LATER X}{1}$ defined uniquely and $\res{\LATER X}{i+1}
  =\res{X}{i}$. Its action on arrows $f:X\to Y$ is $(\LATER f)_1=id_{\{\ast\}}$ and
  $(\LATER f)_{i+1}=f_i$. The natural transformation $\NEXT:id_{\trees}\natto\LATER$
  has $\NEXT_1$ unique and $\NEXT_{i+1}=\res{X}{i}$ for any $X$.
\end{itemize}
\end{definition}
\begin{definition}\label{def:types_denote}
  We interpet types in context $\nabla\vdash A$, where $\nabla$ contains $n$ free
  variables, as functors $\den{\nabla\vdash A}:(\trees^{op}\times\trees)^n\to\trees$,
  usually written $\den{A}$. This mixed variance definition is necessary as variables may
  appear negatively or positively.
\begin{itemize}
\item
  $\den{\nabla,\alpha\vdash\alpha}$ is the projection of the objects or arrows
  corresponding to \emph{positive} occurrences of $\alpha$, e.g. $\den{\alpha}
  (\vec{W},X,Y)=Y$;
\item
  $\den{\ONE}$ and $\den{\NAT}$ are the constant functors
  $\Delta\{\ast\}$ and $\Delta\mathbb{N}$ respectively;
\item
  $\den{A_1\times A_2}(\vec{W})=\den{A_1}(\vec{W})
  \times\den{A_2}(\vec{W})$ and likewise for $\trees$-arrows;
\item
  $\den{A_1\to A_2}(\vec{W})=\den{A_2}
  (\vec{W})^{\den{A_2}(\vec{W}')}$ where $\vec{W}'$ is $\vec{W}$
  with odd and even elements switched to reflect change in polarity, i.e. $(X_1,Y_1,
  \ldots)'=(Y_1,X_1,\ldots)$;
\item
  $\den{\LATER A},\den{\blacksquare A}$ are defined
  by composition with the functors $\LATER,\blacksquare$ (Def.~\ref{def:functors}).
\item 
  $\den{\mu\alpha.A}(\vec{W}) = \mathsf{Fix}(F)$, where
  $F:(\trees^{op}\times\trees)\to\trees$
  is the functor given by $F(X,Y) = \den{A}(\vec{W},X,Y)$ and
  $\mathsf{Fix}(F)$ is the unique (up to isomorphism) $X$ such that
  $F(X,X)\cong X$. The existence of such $X$
  relies on $F$ being a suitably locally contractive functor,
  which follows by Birkedal et al~\cite[Sec.~4.5]{Birkedal-et-al:topos-of-trees}
  and the fact that $\blacksquare$ is only ever applied to closed types.
  This restriction on $\blacksquare$ is necessary because the functor
  $\blacksquare$ is not strong.
\end{itemize}
\end{definition}
\begin{example}\label{ex:denote_streams}
  $\den{\gStream{\NAT}}_i=\mathbb{N}^i$, with projections as restriction functions, so is
  an object of \emph{approximations} of streams -- first the head, then the first two
  elements, and so forth. $\den{\Stream{\NAT}}_i=\mathbb{N}^{\omega}$ at all levels, so is
  the constant object of streams. More generally, any polynomial functor $F$ on $\sets$
  can be assigned a $\lambdanext$-type $A_F$ with a free type variable $\alpha$ that occurs guarded. The
  denotation of $\blacksquare\mu\alpha. A_F$ is the constant object of the carrier of the
  final coalgebra for $F$~\cite[Thm. $2$]{Mogelberg:tt-productive-coprogramming}.
\end{example}
\begin{lemma}\label{lem:rec_types}
  The interpretation of a recursive type is isomorphic to the
  interpretation of its unfolding:
  $\den{\mu\alpha.A}(\vec{W})\iso \den{A[\mu\alpha.A/\alpha]}(\vec{W})$.
\end{lemma}
\begin{lemma}\label{lem:const_types}
Closed constant types denote constant objects in $\trees$.
\end{lemma}

Note that the converse does not apply; for example $\den{\LATER 1}$ is a constant
object.
\begin{definition}\label{def:terms_denote}
We interpret typing contexts $\Gamma=x_1 : A_1, \ldots, x_n :A_n$ as $\trees$-objects
$\den{\Gamma}\defeq \den{A_1}\times\cdots\times\den{A_n}$ and hence interpret
typed terms-in-context $\Gamma\vdash t:A$ as $\trees$-arrows
$\den{\Gamma\vdash t:A}:\den{\Gamma}\to\den{A}$ (usually written $\den{t}$) as
follows.

$\den{x}$ is the projection $\den{\Gamma}\times\den{A}\to\den{A}$. $\den{\ZERO}$
and $\den{\SUCC t}$ are as obvious. Term-formers for products and function spaces
are interpreted via the cartesian closed structure of $\trees$. Exponentials are not
pointwise, so we give explicitly:
\begin{itemize}
  \item $\den{\lambda x.t}_i(\gamma)_j$ maps $a\mapsto
    \den{\Gamma,x:A\vdash t:B}_j(\gamma\hspace{-0.3em}\upharpoonright_j,a)$, 
    where $\gamma\hspace{-0.3em}\upharpoonright_j$ is the result of applying
    restriction functions to $\gamma\in\den{\Gamma}_i$ to get an element of
    $\den{\Gamma}_j$;
  \item $\den{t_1t_2}_i(\gamma)=(\den{t_1}_i(\gamma)_i) \circ
    \den{t_2}_i(\gamma)$;
\end{itemize}
$\den{\FOLD t}$ and $\den{\UNFOLD t}$ are defined via composition with the 
isomorphisms of Lem. \ref{lem:rec_types}.
$\den{\NEXT t}$ and $\den{\UNBOX t}$ are defined by composition with the natural
transformations introduced in Def.~\ref{def:functors}. The final three cases are
\begin{itemize}
\item
  $\den{\PREV [x_1\explsubst t_1,\ldots] .t}_i(\gamma)\defeq\den{t}_{i+1}(\den{t_1}_i
  (\gamma),\ldots)$, where $\den{t_1}_i(\gamma)\in\den{A_1}_i$ is also in
  $\den{A_1}_{i+1}$ by Lem.~\ref{lem:const_types};
\item
  $\den{\BOX [x_1\explsubst t_1,\ldots].t}_i(\gamma)_j=
  \den{t}_j(\den{t_1}_i(\gamma),
  \ldots)$, again using Lem.~\ref{lem:const_types};
\item
  $\den{t_1\APP t_2}_1$ is defined uniquely; $\den{t_1\APP t_2}_{i+1}(\gamma)\defeq
  (\den{t_1}_{i+1}(\gamma)_i)\circ\den{t_2}_{i+1}(\gamma)$.
\end{itemize}
\end{definition}
\begin{lemma}\label{lem:subst}
Given  typed terms in context $x_1:A_1,\ldots,x_m:A_m\vdash t:A$ and $\Gamma
\vdash t_k:A_k$ for $1\leq k\leq m$,
$\den{t[\vec{t}/\vec{x}]}_i(\gamma)=\den{t}_i(\den{t_1}_i(\gamma),\ldots,
\den{t_m}_i(\gamma))$.
\end{lemma}
\begin{theorem}[Soundness]\label{lem:soundness}
If $t\redrt u$ then $\den{t}=\den{u}$.
\end{theorem}

We now define a logical relation between our denotational semantics and
terms, from which both normalisation and adequacy will follow. Doing this inductively
proves rather delicate, because induction on size will not support reasoning
about our values, as $\FOLD$ refers to a larger type in its premise.
This motivates a notion
of \emph{unguarded size} under which $A[\mu\alpha.A/\alpha]$ is `smaller' than $\mu
\alpha.A$. But under this metric $\LATER A$ is smaller than $A$, so $\NEXT$ now
poses a problem. But the meaning of $\LATER A$ at index $i+1$ is determined by
$A$ at index $i$, and so, as in Birkedal et al~\cite{Birkedal:Metric}, our relation will
also induct on
index. This in turn creates problems with $\BOX$, whose meaning refers to all indexes
simultaneously, motivating a notion of \emph{box depth}, allowing us finally to
attain well-defined induction.
\begin{definition}
The \emph{unguarded size} $\usize$ of an open type follows the obvious definition for
type size, except that $\usize(\LATER A)=0$.

The \emph{box depth} $\boxd$ of an open type is
\begin{itemize}
  \item $\boxd(A)=0$ for $A\in\{\alpha,\EMPTY,\ONE,\NAT\}$;
  \item $\boxd(A\times B)=\mathsf{min}(\boxd(A),\boxd(B))$, and similarly for
    $\boxd(A\to B)$;
  \item $\boxd(\mu\alpha.A)=\boxd(A)$, and similarly for $\boxd(\LATER A)$;
  \item $\boxd(\blacksquare A)=\boxd(A)+1$.
\end{itemize}
\end{definition}
\begin{lemma}\label{lem:boxdepth}
\begin{enumerate}[ref={\ref{lem:boxdepth}.(\roman*)}]
\item
  $\alpha$ guarded in $A$ implies $\usize(A[B/\alpha])\leq \usize(A)$.
  \label{lem:boxdepth-i}
\item
  $\boxd(B)\leq \boxd(A)$ implies $\boxd(A[B/\alpha])\leq \boxd(A)$
  \label{lem:boxdepth-ii}
\end{enumerate}
\end{lemma}
\begin{definition}\label{Def:log_rel}
The family of relations $\lrel{i}{A}$, indexed by closed types $A$ and positive integers
$i$, relates elements of the semantics $a\in\den{A}_i$ and closed typed terms $t:A$
and is defined as
\begin{itemize}
\item $\ast \lrel{i}{\ONE} t$ iff $t\redrt\UNIT$;
\item $n \lrel{i}{\NAT} t$ iff $t \redrt \SUCC^n\ZERO$;
\item $(a_1,a_2)\lrel{i}{A_1\times A_2} t$ iff $t\redrt\langle t_1,t_2\rangle$ and
  $a_d\lrel{i}{A_d}t_d$ for $d\in\{1,2\}$;
\item $f\lrel{i}{A\to B} t$ iff $t\redrt\lambda x.s$ and for all $j\leq i$, $a\lrel{j}{A}u$
  implies $f_j(a)\lrel{j}{B} s[u/x]$;
\item $a\lrel{i}{\mu\alpha.A}t$ iff $t\redrt \FOLD u$ and $h_i(a)\lrel{i}{A[\mu
  \alpha.A/\alpha]}u$, where $h$ is the ``unfold'' isomorphism for the recursive type
  (ref. Lem.~\ref{lem:rec_types});
\item $a\lrel{i}{\LATER A} t$ iff $t\redrt\NEXT u$ and, where $i>1$, $a\lrel{i-1}{A}u$.
\item $a\lrel{i}{\blacksquare A} t$ iff $t\redrt\BOX u$ and
  for all $j$, $a_j\lrel{j}{A}u$;
\end{itemize}
This is well-defined by induction on the lexicographic ordering on box depth, then
index, then unguarded size. First the $\blacksquare$ case strictly decreases box depth,
and no other case increases it (ref. Lem.~\ref{lem:boxdepth-ii} for $\mu$-types).
Second the $\LATER$ case strictly decreases index, and no other case increases it
(disregarding $\blacksquare$). Finally all other cases strictly decrease unguarded size,
as seen via Lem.~\ref{lem:boxdepth-i} for $\mu$-types.
\end{definition}
\begin{lemma}[Fundamental Lemma]\label{lem:ad}
Take $\Gamma=(x_1:A_1,\ldots,x_m:A_m)$, $\Gamma\vdash t:A$, and $\vdash t_k:
A_k$ for $1\leq k\leq m$. Then for all $i$, if $a_k\lrel{i}{A_k}t_k$ for all $k$, then
\[
  \den{\Gamma\vdash t:A}_i(\vec{a})\,\lrel{i}{A}\,t[\vec{t}/\vec{x}].
\]
\end{lemma}
\begin{theorem}[Adequacy and Normalisation]
  \label{thm:ad_norm}
\begin{enumerate}[ref={\ref{thm:ad_norm}.(\roman*)}]
  \item For all closed terms $\vdash t:A$ it holds that $\den{t}_i\lrel{i}{A} t$;
  \item $\den{\vdash t:\NAT}_i=n$ implies $t\redrt \SUCC^n\ZERO$;
    \label{thm:ad_norm:den-implies-red}
  \item All closed typed terms evaluate to a value.
\end{enumerate}
\end{theorem}
\begin{proof}
$(i)$ specialises Lem.~\ref{lem:ad} to closed types. $(ii),(iii)$ hold by $(i)$ and
inspection of Def.~\ref{Def:log_rel}.
\end{proof}
\begin{definition}
Typed \emph{contexts} with typed holes are defined as obvious. Two terms
$\Gamma \vdash t : A, \Gamma \vdash u : A$
are \emph{contextually equivalent}, written $t\ceq u$, if for all
\emph{closing} contexts $C$ of type $\NAT$, the terms $C[t]$ and $C[u]$ reduce to the same value.
\end{definition}
\begin{corollary}
  \label{cor:adequacy:den-implies-ctxeq}
  $\den{t}=\den{u}$ implies $t\ceq u$.
\end{corollary}
\begin{proof}
$\den{C[t]}=\den{C[u]}$ by compositionality of the denotational semantics . Then
by Thm.~\ref{thm:ad_norm:den-implies-red} they reduce to the same value.
\end{proof}


%% file: logic.tex
\section{Logic for Guarded Lambda Calculus}
\label{sec:logic}

This section presents our program logic $\logiclambdanext$ for the guarded
$\lambda$-calculus.  The logic is an extension of the internal
language of $\Sl$~\cite{Birkedal-et-al:topos-of-trees,Clouston:Sequent}. Thus it extends multisorted
intuitionistic higher-order logic with two propositional modalities
$\later$ and $\always$, pronounced “later” and “always”
respectively.
The term language of $\logiclambdanext$ includes the terms of
$\lambdanext$, and the types of $\logiclambdanext$
include types definable in $\lambdanext$. We write
$\Omega$ for the type of propositions, and also for the subobject
classifier of $\Sl$.

The rules for \emph{definitional equality} extend the usual $\beta\eta$-laws
for functions and products with new equations for the
new $\lambdanext$ constructs, listed in Fig.~\ref{fig:additional-equations}.
%
\begin{figure}[htb]
  \centering
  \begin{mathpar}
    \inferrule*{\Gamma \vdash t : A\left[\mu \alpha . A\right/\alpha]}
    {\Gamma \vdash \UNFOLD (\FOLD t) = t}
    \and
    \inferrule*{\Gamma \vdash t : \mu \alpha . A}
    {\Gamma \vdash \FOLD (\UNFOLD t) = t}
    \and
    \inferrule*{\Gamma \vdash t_1 : A \to B\\
      \Gamma \vdash t_2 : A}{\Gamma \vdash \NEXT t_1 \APP \NEXT t_2 = \NEXT (t_1 t_2)}
    \and
    \inferrule*{%
      \Gamma_{\blacksquare} \vdash t: A \\
      \Gamma\vdash \vec{t} : \Gamma_{\blacksquare}}{%
      \Gamma \vdash \PREV[\vec{x}\explsubst\vec{t}].(\NEXT t) = t\left[\vec{t}/\vec{x}\right]}
    \and
    \inferrule*{%
      \Gamma_{\blacksquare} \vdash t: \LATER A \\
      \Gamma\vdash \vec{t} : \Gamma_{\blacksquare}}{%
      \Gamma \vdash \NEXT\left(\PREV[\vec{x}\explsubst\vec{t}].t\right) = t\left[\vec{t}/\vec{x}\right]}
    \and
      \inferrule*{%
      \Gamma_{\blacksquare} \vdash t:A \\
      \Gamma\vdash \vec{t} : \Gamma_{\blacksquare}}{%
      \Gamma \vdash \UNBOX(\BOX[\vec{x}\explsubst\vec{t}].t) = t\left[\vec{t}/{\vec{x}}\right]}
    \and
      \inferrule*{%
      \Gamma_{\blacksquare}\vdash t:\blacksquare A \\
      \Gamma\vdash \vec{t} : \Gamma_{\blacksquare}}{%
      \Gamma \vdash \BOX[\vec{x}\explsubst\vec{t}].\UNBOX t = t\left[\vec{t}/{\vec{x}}\right]}
  \end{mathpar}
  \caption{Additional equations. The context $\Gamma_{\blacksquare}$ is assumed constant.}
\label{fig:additional-equations}
\end{figure}

\begin{definition}
  \label{def:total-types}
  A type $X$ is \emph{total and inhabited} if the formula $\total{X}\equiv
  \forall x : \LATER  X, \exists x' : X, \nxt(x') =_{\LATER X} x$ is valid.
\end{definition}

All of the $\lambdanext$-types defined in Sec.~\ref{sec:calculus} are total and
inhabited (see  
\begin{cameraversion}
the extended version~\cite{ARXIVVERSION}
\end{cameraversion}
\begin{arxivversion}
App.~\ref{sec:about-totality-types}
\end{arxivversion}
for a proof using the semantics of
the logic), but that is not the case when we include sum types
as the empty type is not inhabited.

Corresponding to the modalities $\LATER$ and $\blacksquare$ on types, we have modalities
$\later$ and $\always$ on formulas. The modality $\later$ is used to
express that a formula holds only ``later'', that is, after a time step.  It is given by a
function symbol $\later : \Omega \to \Omega$. The $\always$ modality is used to express that a formula holds for all time
steps. Unlike the $\later$ modality, $\always$ on formulas does not arise from a
function on $\Omega$~\cite{Bizjak-et-al:countable-nondet-internal}.
As with $\BOX$, it is only well-behaved in
constant contexts, so we will only allow $\always$ in such contexts. The rules for
$\later$ and $\always$ are listed in Fig.~\ref{fig:old-rules}. 

\begin{figure}[htb]
  \centering
  \begin{mathpar}
    \inferrule*[Right=L\"ob]{ }{\Gamma \mid \Xi, (\later \phi \implies \phi) \vdash \phi}
    \and
    \inferrule*[Right=$\exists\later$]{ }{\Gamma, x : X \mid \exists y : Y, \later
      \phi(x,y) \vdash \later \left(\exists y : Y, \phi(x, y)\right)}
    \and
    \inferrule*[Right=$\forall\later$]{ }{\Gamma, x:X \mid \later (\forall y : Y, \phi(x, y)) \vdash \forall y :
      Y, \later \phi(x, y)}
    \and
    \inferrule*{ }{\Gamma \mid \Xi, \phi \vdash \later \phi}
    \and
    \inferrule*{\star \in \{\land, \lor, \implies\} }
    {\Gamma \mid \later (\phi \star \psi) \dashv\vdash \later\phi \star \later\psi}
    \and
    \inferrule*{\Gamma \mid \lnot\lnot \phi \vdash \psi}{\Gamma \mid \phi \vdash \always \psi}
    \and
    \inferrule*{\Gamma \mid \phi \vdash \always \psi}{\Gamma \mid \lnot\lnot \phi \vdash \psi}
    \and
    \inferrule*{\Gamma \mid \phi \vdash \psi}{\Gamma \mid \always \phi \vdash \always \psi}
    \and
    \inferrule*{ }{\Gamma \mid \always \phi \vdash \phi}
    \and \inferrule*{ }{\Gamma \mid \always \phi \vdash \always\always \phi}
    \and
    \inferrule*[Right=\eqlaternextrule]{ }{\forall x, y : X . \later (x =_X y) \iff \NEXT x =_{\LATER X} \NEXT y}
    \phantom{\eqlaternextrule}
  \end{mathpar}
  \caption{Rules for $\later$ and $\always$. The judgement $\Gamma \mid \Xi \vdash \phi$
    expresses that in typing context $\Gamma$, hypotheses in $\Xi$ prove $\phi$.  The
    converse entailment in $\forall\later$ and $\exists\later$ rules holds if $Y$ is
    \emph{total and inhabited}. In all rules involving the $\always$ the context $\Gamma$
    is assumed constant.}
  \label{fig:old-rules}
\end{figure}

The $\later$ modality can in fact be defined in terms of $\lift : \LATER \Omega \to
\Omega$ (called $succ$ by Birkedal et al~\cite{Birkedal-et-al:topos-of-trees})
as $\later = \lift \comp \NEXT$. The $\lift$ function will be useful since it allows us to
define predicates over guarded types, such as predicates on $\gStream{\NN}$.

The semantics of the logic is given in $\Sl$; terms are interpreted as
morphisms of $\Sl$ and formulas are interpreted via the subobject classifier.
We do not present the semantics here; except for the
new terms of $\lambdanext$, whose semantics are defined in
Sec.~\ref{sec:denot}, the semantics are as
in~\cite{Birkedal-et-al:topos-of-trees,Bizjak-et-al:countable-nondet-internal}.

Later we will come to the problem of proving $x=_{\blacksquare A}y$ from $\UNBOX x
=_A\UNBOX y$, where $x,y$ have type $\blacksquare A$. This in general
does not hold, but using the semantics of $\logiclambdanext$ we can prove
the proposition below.
\begin{proposition}
  \label{prop:always-unbox-injective}
  The formula $\always(\UNBOX x =_A \UNBOX y) \implies x =_{\blacksquare A} y$ is valid.
\end{proposition}

There exists a fixed-point combinator of type $(\LATER A \to A) \to A$
for all types $A$ in the logic (not only those of in $\lambdanext$)
\cite[Thm. $2.4$]{Birkedal-et-al:topos-of-trees}; we also 
write $\Theta$ for it.
\begin{proposition}
  \label{prop:theta-is-a-fp}
  For any 
  term $f : \LATER A \to A$ we have $\Theta f =_{A} f \left(\NEXT (\Theta
    f)\right)$ and, if $u$ is any other term such that $f (\NEXT u) =_{A} u$, then $u =_{A}
  \Theta f$.
\end{proposition}
In particular this can be used for recursive definitions of
predicates. For instance if $P:\NN\to\Omega$ is a predicate on natural numbers we
can define a predicate $\liftstr{P}$ on $\gStream{\NN}$
expressing that $P$ holds for all elements of the stream:
\begin{align*}
  \liftstr{P} \defeq
  \Theta
  \lambda r . \lambda xs . P (\head xs) \land \lift \left(r \APP (\tail xs)\right) :
  \gStream{\NN} \to \Omega.
\end{align*}
The logic may be used to prove contextual equivalence of programs:
\begin{theorem}
  \label{thm:prop-equality-ctx-equiv}
  Let $t_1$ and $t_2$ be two $\lambdanext$ terms of type $A$ in context $\Gamma$. If
  the sequent $\Gamma \mid \emptyset \vdash t_1 =_{A} t_2$
  is provable then $t_1$ and $t_2$ are contextually equivalent.
\end{theorem}
\begin{proof}
  Recall that equality in the internal logic of a topos is just equality of
  morphisms. Hence $t_1$ and $t_2$ denote same morphism from $\Gamma$ to
  $A$. Adequacy (Cor.~\ref{cor:adequacy:den-implies-ctxeq}) then implies that $t_1$
  and $t_2$ are contextually equivalent.
\end{proof}

\begin{example}
  \label{ex:example-proofs}
  We list some properties provable using the logic. Except for the
  first property all proof details are in
\begin{cameraversion}
the extended version~\cite{ARXIVVERSION}.
\end{cameraversion}
\begin{arxivversion}
App.~\ref{sec:example-proofs}.
\end{arxivversion}
  \begin{enumerate}[ref={\ref{ex:example-proofs}.(\roman*)}]
  \item For any $f : A \to B$ and $g : B \to C$ we have
    \[(\map f) \comp (\map g) =_{\gStream{A} \to \gStream{C}} \map (f \comp g).\]
    Unfolding the definition of $\map$ from Ex.~\ref{ex:programs}$(vi)$ and using
    $\beta$-rules and Prop.~\ref{prop:theta-is-a-fp}
    we have  $\map f\,xs = f\,(\head xs) \consin (\NEXT (\map f) \APP (\tail xs))$.
    Equality of functions is extensional so we have to prove
    \begin{align*}
      \Phi \defeq \forall xs : \gStream{A}, \map f\, (\map g\, xs) =_{\gStream{C}} \map (f \comp g)\,xs.
    \end{align*}
    The proof is by L\"ob induction, so we assume $\later\Phi$ and take $xs :
    \gStream{A}$. Using the above property of $\map$ we unfold $\map f\, (\map g\, xs)$
    to
    \begin{align*}
      f\,(g\,(\head xs)) \consin \left(\NEXT (\map f) \APP ((\NEXT (\map g)) \APP \tail xs)\right)
    \end{align*}
    and we unfold $\map (f\comp g)\,xs$ to 
    $f\,(g\,(\head xs)) \consin \left(\NEXT (\map (f\comp g)) \APP \tail xs\right)$.
    Since $\gStream{A}$ is a total type there is a $xs' : \gStream{A}$ such that $\NEXT xs'
    = \tail xs$. Using this and the rule for $\APP$ we have
    \begin{align*}
      \NEXT (\map f) \APP ((\NEXT (\map g)) \APP \tail xs) =_{\LATER \gStream{C}} \NEXT
      (\map f (\map g\,xs'))
    \end{align*}
    and $\NEXT (\map (f\comp g)) \APP \tail xs =_{\LATER \gStream{C}} \NEXT (\map (f \comp g)\,xs')$.
    From the induction hypothesis $\later\Phi$ we have 
    $\later (\map (f \comp g)\,xs' =_{\gStream{C}} \map f\,(\map g\,xs'))$
    and so rule $\eqlaternextrule$ concludes the proof.
  \item We can also reason about acausal functions. For any $n : \NAT,
    f:\NAT\to\NAT$,
    \[\everysecond (\BOX \iota . \iterate\,(\NEXT f)\, n)
      =_{\gStream{\NAT}} \iterate\,(\NEXT f^2)\, n,\]
    where $f^2$ is $\lambda m.f\, (f\, m)$.
    The proof again uses L\"ob induction.
  \item Since our logic is higher-order we can state and prove very general properties,
    for instance the following general property of map
    \begin{align*}
      &\forall P, Q : (\NN \to \Omega), \forall f : \NN \to \NN,
      (\forall x : \NN, P(x) \implies Q(f(x)))\\
      &\implies\forall xs : \gStream{\NN}, \liftstr{P}(xs) \implies \liftstr{Q}(\map f\,xs).
    \end{align*}
    The proof illustrates the use of the property $\lift \comp \NEXT = \later$.
  \item \label{ex:unbox-equality}
    Given a closed term (we can generalise to terms in constant contexts) $f$ of type $A \to
    B$ we have $\BOX f$ of type $\blacksquare (A \to B)$. Define $\mathcal{L}(f) = \lim
    (\BOX f)$ of type $\blacksquare A \to \blacksquare B$.
    For any closed term $f : A \to B$ and $x : \blacksquare A$ we can then prove
    $\UNBOX (\mathcal{L}(f)\,x) =_{B} f\,(\UNBOX x)$. Then using 
    Prop.~\ref{prop:always-unbox-injective} we can, for instance,
    prove $\mathcal{L}(f \comp g) = \mathcal{L}(f) \comp \mathcal{L}(g)$.

    For functions of arity $k$ we define $\mathcal{L}_k$ using $\mathcal{L}$, and analogous
    properties hold, e.g. we have $\UNBOX(\mathcal{L}_2(f)\,x\,y) = f\,(\UNBOX x)\,(\UNBOX
    y)$, which allows us to transfer equalities proved for functions on guarded types to
    functions on $\blacksquare$'d types; see Sec.~\ref{sec:definable-functions} for an
    example. 
  \end{enumerate}
\end{example}




%% file: bde-examples.tex
\section{Behavioural Differential Equations in $\lambdanext$}
\label{sec:definable-functions}

In this section we demonstrate the expressivity of our approach by
showing how to construct solutions to behavioural
differential equations~\cite{Rutten:2003:bde} in $\lambdanext$, and
how to reason about such functions in $\logiclambdanext$, rather than with bisimulation
as is more traditional.  These ideas are best explained via a simple example.

Supposing addition $+ : \NN \to \NN \to \NN$ is
given, then pointwise addition of streams, $\limplus$, can be defined by the
following behavioural differential equation
\begin{align*}
  \limhead (\limplus \sigma_1\, \sigma_2) = \limhead \sigma_1 + \limhead \sigma_2
  \qquad\qquad
  \limtail (\limplus \sigma_1\, \sigma_2) = \limplus (\limtail \sigma_1)\, (\limtail \sigma_2).
\end{align*}
To define the solution to this behavioural differential equation in $\lambdanext$, we
first translate it to a function on guarded streams
$\plus : \gStream{\NAT} \to \gStream{\NAT} \to \gStream{\NAT}$, as
\begin{align*}
   \plus \defeq 
   \Theta \lambda f . \lambda s_1 . \lambda s_2 . (\head s_1 + \head s_2)\consin(f \APP (\tail  s_1) \APP (\tail s_2))
\end{align*}
then define $\limplus : \Stream{\NAT}\to \Stream{\NAT} \to \Stream{\NAT}$ by
$\limplus = \mathcal{L}_2(\plus)$.
By Prop.~\ref{prop:theta-is-a-fp} we have
\begin{align}
  \label{eq:plus-defining-eq}
 \plus = \lambda s_1 . \lambda s_2 . (\head s_1 + \head s_2)\consin ((\NEXT\plus) \APP (\tail s_1)
  \APP (\tail s_2)).
\end{align}
This definition of $\limplus$ satisfies the
specification given by the behavioural differential equation above.
Let $\sigma_1, \sigma_2: \Stream{\NAT}$ and recall that
$\limhead = \head \comp \lambda s . \UNBOX s$.
Then use Ex.~\ref{ex:unbox-equality} and 
equality \eqref{eq:plus-defining-eq} to get
$\limhead (\limplus \sigma_1 \sigma_2) = \limhead \sigma_1 + \limhead \sigma_2$.

For $\limtail$
we proceed similarly, also using that $\tail (\UNBOX \sigma) = \NEXT (\UNBOX
(\limtail \sigma))$ which can be proved using the $\beta$-rule for $\BOX$ and the
$\eta$-rule for $\NEXT$.

Since $\plus$ is defined via guarded recursion we can reason about it with L\"ob
induction, for example to prove that it is commutative.
Ex.~\ref{ex:unbox-equality} and Prop.~\ref{prop:always-unbox-injective}
then immediately give that $\limplus$ on \emph{coinductive} streams $\Stream{\NAT}$ is
commutative.

Once we have defined $\plus$ we can use it when defining other functions on streams, for
instance stream multiplication $\otimes$ which is specified by equations
\begin{align*}
  \begin{split}
    \limhead (\sigma_1 \otimes \sigma_2) &= (\limhead \sigma_1) \cdot (\limhead \sigma_2)
  \end{split}\hspace{-1mm}
  \begin{split}
    \limtail (\sigma_1 \otimes \sigma_2) &=
    \left(\rho (\limhead \sigma_1) \otimes (\limtail \sigma_2)\right)
    \oplus \left((\limtail \sigma_1) \otimes \sigma_2\right)
  \end{split}
\end{align*}
where $\rho(n)$ is a stream with head $n$ and tail a stream of zeros, and $\cdot$ is
multiplication of natural numbers, and using $\oplus$ as infix notation for
$\limplus$. We can define $\guarded{\otimes} : \gStream{\NAT} \to \gStream{\NAT} \to
\gStream{\NAT}$ by $\guarded{\otimes} \defeq\mbox{}$ 
\begin{align*}
  \Theta
  \lambda f . \lambda s_1 . &\lambda s_2 .
  \left((\head s_1) \cdot (\head s_2)\right)\consin\\
        &\left(\NEXT\plus 
               \APP(f\APP\NEXT\guarded\iota(\head s_1)\APP\tail s_2)
               \APP(f\APP\tail s_1\APP \NEXT s_2)\right)
\end{align*}
then define $\otimes = \mathcal{L}_2\left(\guarded{\otimes}\right)$. It can be shown
that the function $\otimes$ so defined satisfies
the two defining equations above. Note that the guarded
$\plus$ is used to define $\guarded{\otimes}$, so our approach is
\emph{modular} in the sense of~\cite{milius:abstract-gsos}.

The example above generalises, as we can show that any solution to a
behavioural differential equation in $\sets$ can be obtained via
guarded recursion together with $\mathcal{L}_k$. The formal
statement is somewhat technical
and can be found in 
\begin{cameraversion}
the extended version~\cite{ARXIVVERSION}.
\end{cameraversion}
\begin{arxivversion}
App.~\ref{sec:proof-bde}.
\end{arxivversion}


%% file: conclusion.tex

\section{Discussion}\label{sec:conclusion}

Following Nakano~\cite{Nakano:Modality}, the $\LATER$ modality has been used as
type-former for a number of $\lambda$-calculi for guarded recursion. Nakano's
calculus and some
successors~\cite{Krishnaswami:Ultrametric,Severi:Pure,Abel:Formalized}
permit only \emph{causal} functions. The closest such work to ours is
that of Abel and Vezzosi~\cite{Abel:Formalized}, but due to a lack of destructor for
$\LATER$ their (strong) normalisation result relies on a somewhat artificial operational
semantics where the number of $\NEXT$s that can be reduced under is bounded
by some fixed natural number.

Atkey and McBride's extension of such calculi to acausal
functions~\cite{Atkey:Productive} forms the basis of this paper. We build on
their work by (aside from various minor changes such as eliminating the need to
work modulo first-class type isomorphisms) introducing normalising operational
semantics, an adequacy proof with respect to the topos of trees, and a program logic.

An alterative approach to type-based productivity guarantees are \emph{sized types},
introduced by Hughes et al~\cite{Hughes:Proving} and now extensively developed, for
example integrated into a variant of System $F_{\omega}$~\cite{Abel:Wellfounded}.
Our approach offers some advantages, such as adequate denotational semantics, and a
notion of program proof without appeal to dependent types, but extensions with
realistic language features (e.g. following
M{\o}gelberg~\cite{Mogelberg:tt-productive-coprogramming}) clearly need to be investigated.

%% file: appx_lang_proofs.tex
\section{Proofs for Section~\ref{sec:denot}}\label{app:lang_proofs}

\begin{proof}[of Lem.~\ref{lem:const_types}]
By induction on type formation, with $\LATER A$ case omitted, $\blacksquare A$ a base
case, and $\mu\alpha.A$ considered only where $\alpha$ is not free in $A$.
\end{proof}

\begin{proof}[of Lem.~\ref{lem:subst}]
By induction on the typing of $t$. We present the cases particular to our calculus.

$\NEXT t$: case $i=1$ is trivial.
$\den{\NEXT t[\vec{t}/\vec{x}]}_{i+1}(\gamma)=
\res{\den{A}}{i}\circ\den{t[\vec{t}/\vec{x}]}_{i+1}(\gamma)=
\res{\den{A}}{i}\circ\den{t}_{i+1}(\den{t_1}_{i+1}(\gamma),\ldots)$ by induction,
which is $\den{\NEXT t}_{i+1}(\den{t_1}_{i+1}(\gamma),\ldots)$.

$\den{(\PREV[\vec{y}\explsubst\vec{u}].t)[\vec{t}/\vec{x}]}_i(\gamma)=
\den{\PREV[\vec{y}\explsubst\vec{u}[\vec{t}/\vec{x}]].t}_i(\gamma)$, which by definition
is $\den{t}_{i+1}(\den{u_1[\vec{t}/\vec{x}]}_i(\gamma),\ldots)=
\den{t}_{i+1}(\den{u_1}_i(\den{t_1}_i(\gamma),\ldots),\ldots)$ by induction, which is
$\den{\PREV[\vec{y}\explsubst\vec{u}].t}_i(\den{t_1}_i(\gamma),\ldots)$.

$\den{\BOX[\vec{y}\explsubst\vec{u}[\vec{t}/\vec{x}]].t}_i(\gamma)_j=
\den{t}_j(\den{u_1[\vec{t}/\vec{x}]}_i(\gamma),\ldots)$, which by induction equals
$\den{t}_j(\den{u_1}_i(\den{t_1}_i(\gamma),\ldots),\ldots)=
\den{\BOX[\vec{y}\explsubst\vec{u}].t}_i(\den{t_1}_i(\gamma),\ldots)_j$.

$\den{\UNBOX t[\vec{t}/\vec{x}]}_i(\gamma)=
\den{t[\vec{t}/\vec{x}]}_i(\gamma)_i=
\den{t}_i(\den{t_1}_i(\gamma),\ldots)_i$ by induction, which is
$\den{\UNBOX t}_i(\den{t_1}_i(\gamma),\ldots)$.

$u_1\APP u_2$: case $i=1$ is trivial.
$\den{(u_1\APP u_2)[\vec{t}/\vec{x}]}_{i+1}(\gamma)=
(\den{u_1[\vec{t}/\vec{x}]}_{i+1}(\gamma)_i)\circ
\den{u_2[\vec{t}/\vec{x}]}_{i+1}(\gamma)=
(\den{u_1}_{i+1}(\den{t_1}_{i+1}(\gamma),\ldots)_i)\circ
\den{u_2}_{i+1}(\den{t_1}_{i+1}(\gamma),\ldots)$, which is
$\den{u_1\APP u_2}_{i+1}(\den{t_1}_{i+1}(\gamma),\ldots)$.
\end{proof}

\begin{proof}[of Soundness Thm.~\ref{lem:soundness}]
We verify the reduction rules of Def.~\ref{def:redrule}; extending this to any
evaluation context, and to $\redrt$, is easy. The product reduction case is standard,
and function case requires Lem.~\ref{lem:subst}. $\UNFOLD\FOLD$ is the application of
mutually inverse arrows.

$\den{\PREV[\vec{x}\explsubst\vec{t}].t}_i=\den{t}_{i+1}(\den{t_1}_i,\ldots)$. Each
$t_k$ is closed, so is denoted by an arrow from $1$ to the constant $\trees$-object
$\den{A_k}$, so by naturality $\den{t_k}_i=\den{t_k}_{i+1}$. But $\den{t}_{i+1}
(\den{t_1}_{i+1},\ldots)=\den{t[\vec{t}/\vec{x}]}_{i+1}$ by Lem.~\ref{lem:subst},
which is $\den{\PREV t[\vec{t}/\vec{x}]}_i$.

$\den{\PREV\NEXT t}_i=\den{\NEXT t}_{i+1}=\den{t}_i$.

$\den{\UNBOX(\BOX[\vec{x}\explsubst\vec{t}].t)}_i=(\den{\BOX
  [\vec{x}\explsubst\vec{t}].t}_i)_i
=\den{t}_i(\den{t_1}_i,\ldots)=\den{t[\vec{t}/\vec{x}]}_i$.

With $\APP$-reduction, index $1$ is trivial.
$\den{\NEXT t_1\APP\NEXT t_2}_{i+1}=
(\den{\NEXT t_1}_{i+1})_i\circ\den{\NEXT t_2}_{i+1}=
(\res{\den{A\to B}}{i}\circ \den{t_1}_{i+1})_i\circ
\res{\den{A}}{i}\circ \den{t_2}_{i+1}=
(\den{t_1}_i\circ\res{1}{i})_i\circ\den{t_2}_i\circ \res{1}{i}$ by naturality, which is
$(\den{t_1}_i)_i\circ\den{t_2}_i=
\den{t_1t_2}_i=
\den{t_1t_2}_i\circ\res{1}{i}=
\res{\den{B}}{i}\circ \den{t_1t_2}_{i+1}=
\den{\NEXT(t_1t_2)}_{i+1}$.
\end{proof}

\begin{proof}[of Lem.~\ref{lem:boxdepth}]
By induction on the construction of the type $A$.

$(i)$ follows with only interesting case the variable case -- $A$ cannot be $\alpha$
because of the requirement that $\alpha$ be guarded in $A$.

$(ii)$ follows with interesting cases: variable case enforces $bd(B)=0$; binary
type-formers $\times,\to$ have for example $\boxd(A_d)\geq \boxd(A_1\times A_2)$,
so $\boxd(A_d)\geq bd(B)$ and the induction follows; $\blacksquare A$ by construction
has no free variables.
\end{proof}

\begin{lemma}\label{lem:lrel_and_red}
If $t\redrt u$ and $a\lrel{i}{A}u$ then $a\lrel{i}{A}t$.
\end{lemma}
\begin{proof}
All cases follow similarly; consider $A_1\times A_2$. $(a_1,a_2)
\lrel{i}{A_1\times A_2}u$ implies $u\redrt\langle t_1,t_2\rangle$, where this value
obeys some property. But then $t\redrt\langle t_1,t_2\rangle$ similarly.
\end{proof}

\begin{lemma}\label{lem:res_and_rel}
$a\lrel{i+1}{A}t$ implies $\res{\den{A}}{i}(a)\lrel{i}{A}t$.
\end{lemma}
\begin{proof}
Cases $\ONE,\NAT$ are trivial. Case $\times$ follows by induction because
restrictions are defined pointwise. Case $\mu$ follows by induction and the naturality
of the isomorphism $h$. Case $\blacksquare A$ follows because
$\res{\den{\blacksquare A}}{i}(a)=a$.

For $A\to B$ take $j\leq i$ and $a'\lrel{j}{A}u$. By the downwards closure in the
definition of $\lrel{i+1}{A\to B}$ we have $f_j(a')\lrel{j}{B}s[u/x]$. But
$f_j=(\res{\den{A\to B}}{i}(f))_j$.

With $\LATER A$, case $i=1$ is trivial, so take $i=j+1$. $a\lrel{j+2}{\LATER A}t$
means $t\redrt\NEXT u$ and $a\lrel{j+1}{A}u$, so by induction $\res{\den{A}}{j}(a)
\lrel{j}{A}u$, so $\res{\den{\LATER A}}{j+1}(a)\lrel{j}{A}u$ as required.
\end{proof}

\begin{lemma}\label{lem:constrel}
If $a\lrel{i}{A}t$ and $A$ is constant, then $a\lrel{j}{A}t$ for all $j$.
\end{lemma}
\begin{proof}
Easy induction on types, ignoring $\LATER A$ and treating $\blacksquare A$ as a base
case.
\end{proof}

We finally turn to the proof of the Fundamental Lemma.

\begin{proof}[of Lem.~\ref{lem:ad}]
By induction on the typing $\Gamma\vdash t:A$. $\UNIT,\ZERO$ cases are trivial, and
$\langle u_1,u_2\rangle,\FOLD t$ cases follow by easy induction.

$\SUCC$: If $t[\vec{t}/\vec{x}]$ reduces to $\SUCC^l\ZERO$ for some $l$ then
$\SUCC t[\vec{t}/\vec{x}]$ reduces to $\SUCC^{l+1}\ZERO$, as we may reduce
under the $\SUCC$.

$\pi_d t$: If $\den{t}_i(\vec{a})\lrel{i}{A_1\times A_2}t[\vec{t}/\vec{x}]$ then $t
[\vec{t}/\vec{x}]\redrt\langle u_1,u_2\rangle$ and $u_d$ is related to the $d$'th projection of $\den{t}_i(\vec{a})$. But then $\pi_d t[\vec{t}/\vec{x}]\redrt\pi_d\langle
u_1,u_2\rangle\red u_d$, so Lem.~\ref{lem:lrel_and_red} completes the case.

$\lambda x.t$: Taking $j\leq i$ and $a\lrel{j}{A}u$, we must show that
$\den{\lambda x.t}_i(\vec{a})_j(a)\lrel{j}{B}t[\vec{t}/\vec{x}][u/x]$. The left hand
side is $\den{t}_j(\vec{a}\hspace{-0.3em}\upharpoonright_j,a)$. For each $k$, $a_k
\hspace{-0.3em}\upharpoonright_j\hspace{-0.3em}\lrel{j}{A_k}t_k$ by
Lem.~\ref{lem:res_and_rel}, and induction completes the case.

$u_1u_2$: By induction $u_1[\vec{t}/\vec{x}]\redrt \lambda x.s$ and
$\den{u_1}_k(\vec{a})_k(\den{u_2}_k(\vec{a}))\lrel{i}{B}s[u_2[\vec{t}/\vec{x}]/x]$.
Now $(u_1u_2)\redrt(\lambda x.s)(u_2[\vec{t}/\vec{x}])\red
s[u_2[\vec{t}/\vec{x}]/x]$, and Lem.~\ref{lem:lrel_and_red} completes.

$\UNFOLD t$: we reduce under $\UNFOLD$,
then reduce $\UNFOLD\FOLD$, then use Lem.~\ref{lem:lrel_and_red}.

$\NEXT t$: Trivial for index $1$. For $i=j+1$, if each $a_k\lrel{j+1}{A_k}t_k$ then by
Lem.~\ref{lem:res_and_rel} $\res{\den{A_k}}{j}(a_k)\lrel{j}{A_k}
t_k$. Then by induction $\den{t}_j\circ\res{\den{A_1}\times\cdots\den{A_m}}{j}(\vec{a})\lrel{j}{A}t[\vec{t}/\vec{x}]$, whose left side is by naturality
$\res{\den{A}}{j}\circ\den{t}_{j+1}(\vec{a})=\den{\NEXT t}_{j+1}(\vec{a})$.

$\PREV[\vec{y}\explsubst\vec{u}].t$: $\den{u_k}_i(\vec{a})\lrel{i}{A_k}u_k[\vec{t}/
\vec{x}]$ by induction, so $\den{u_k}_i(\vec{a})\lrel{i+1}{A_k}u_k[\vec{t}/
\vec{x}]$ by Lem.~\ref{lem:constrel}. Then $\den{t}_{i+1}(\den{u_1}_i(\vec{a}),
\ldots)\lrel{i+1}{\LATER A} t[u_1[\vec{t}/\vec{x}]/y_1,\ldots]$ by induction, so we
have $t[u_1[\vec{t}/\vec{x}]/y_1,\ldots]\redrt\NEXT s$ with $\den{t}_{i+1}
(\den{u_1}_k(\vec{a}),\ldots)\lrel{i}{A}s$. The left hand side is $\den{\PREV[\vec{y}\explsubst\vec{u}] 
.t}_i(\vec{a})$, while
$\PREV[\vec{y}\explsubst\vec{u}[\vec{t}/\vec{x}]].t\red \PREV
t[u_1[\vec{t}/\vec{x}]/y_1,\ldots]
\redrt\PREV\NEXT s\red s$, so Lem.~\ref{lem:lrel_and_red} completes.

$\BOX[\vec{y}\explsubst\vec{u}].t$: To show $\den{\BOX[\vec{y}\explsubst\vec{u}].t}_i
(\vec{a})\lrel{i}{\blacksquare A}\BOX[\vec{y}\explsubst\vec{u}].t)[\vec{t}/\vec{x}]$, we
observe that the right hand side reduces in one step to $\BOX t[u_1[\vec{t}/
\vec{x}]/y_1,\ldots]$. The $j$'th element of the left hand side is $\den{t}_j
(\den{u_1}_k(\vec{a}),\ldots)$. We need to show this is related by $\lrel{j}{A}$ to
$ t[u_1[\vec{t}/\vec{x}]/y_1,\ldots]$; this follows by Lem.~\ref{lem:constrel} and
induction.

$\UNBOX t$: By induction $t[\vec{t}/\vec{x}]\redrt\BOX u$, so $\UNBOX t[\vec{t}/
\vec{x}]\redrt\UNBOX\BOX u\red u$. By induction $\den{t}_i(\vec{a})_i\lrel{i}{A}u$,
so $\den{\UNBOX t}_i(\vec{a})\lrel{i}{A}u$, and Lem.~\ref{lem:lrel_and_red}
completes.

$u_1\APP u_2$: Index $1$ is trivial so set $i=j+1$. $\den{u_2}_{j+1}(\vec{a})
\lrel{j+1}{\LATER A}u_2[\vec{t}/\vec{x}]$ implies $u_2[\vec{t}/\vec{x}]\redrt\NEXT s_2$ with
$\den{u_2}_{j+1}(\vec{a})\lrel{j}{A}s_2$. Similarly $u_1\redrt\NEXT s_1$ and
$s_1\redrt \lambda x.s$ with
$(\den{u_1}_{j+1}(\vec{a})_j)\circ\den{u_2}_{j+1}(\vec{a})\lrel{j}{B}s[s_2/x]$.
The left hand side is exactly $\den{u_1\APP u_2}_{j+1}(\vec{a})$. Now $u_1\APP u_2
\redrt \NEXT s_1\APP u_2\redrt \NEXT s_1\APP \NEXT s_2\red \NEXT(s_1s_2)$, and
$s_1s_2\redrt (\lambda x.s)s_2\red s[s_2/x]$, completing the proof.
\end{proof}


%% file: appx_example_proofs.tex
\section{Example Proofs in \logiclambdanext}
\label{sec:example-proofs}

We first record a substitution property of $\BOX$ and $\PREV$ for later use.
\begin{lemma}
  \label{lem:box-prev-const-subst}
  Let $A_1, \dots, A_k$ and $B$ be constant types and $C$ any type.
  If we have $x : B \vdash t : C$ and $y_1 : A_k, \dots, y_k : A_k \vdash t' : B$ then
  \[
    \BOX\,[x\explsubst t'] . t =_{\blacksquare C} \BOX \iota . t[t'/x].
  \]
  If $C = \LATER D$ then we also have
  \[
    \PREV \,[x\explsubst t']. t =_D \PREV \iota . t[t'/x]
  \]
\end{lemma}
We can prove the first part of the lemma in the logic, using
Prop.~\ref{prop:always-unbox-injective} and the $\beta$-rule for $\BOX$. We can also
prove the second part of the lemma for \emph{total and inhabited types $D$} with the rules
we have stated so far using the $\beta$-rule for $\NEXT$. For arbitrary $D$ we can prove
the lemma using the semantics.

\subsection{Acausal Example}
\label{sec:acausal-example}

To see that L\"ob induction can be used to prove properties of recursively defined
acausal functions we show that for any $n : \NAT$ and any $f : \NAT \to \NAT$ we have
\[
  \everysecond\, (\BOX \iota . \iterate\,(\NEXT f)\, n ) =_{\gStream{\NAT}}
  \iterate\,(\NEXT f^2)\, n,
\]
where we write $f^2$ for $\lambda n . f (f n)$. We first derive the intermediate
result
\begin{equation}
  \label{eq:iterate-tail}
  \forall m : \NAT, \limtail \,(\BOX \iota . \iterate\,(\NEXT f)\, m) =_{\Stream{\NAT}} 
  \BOX \iota . \iterate \, (\NEXT f)\, (f \,m),
\end{equation}
by unfolding and applying Prop.~\ref{prop:theta-is-a-fp}:
\begin{align*}
  \limtail\, (\BOX \iota . \iterate\,(\NEXT f)\, m) &=
  \BOX\, [s \explsubst \BOX \iota . \iterate\,(\NEXT f)\, m] . \PREV \iota . \tail
  (\UNBOX s) \\
  &= \BOX \iota . \PREV \iota . \tail (\iterate\, (\NEXT f)\, m)
    \tag*{(by Lem.~$\ref{lem:box-prev-const-subst}$) } \\
  &= \BOX \iota . \PREV \iota . \NEXT\, (\iterate\, (\NEXT f)\, (f\, m)) \\
  &= \BOX \iota . \iterate\, (\NEXT f)\, (f \, m).
\end{align*}
Now assume
\begin{equation}
  \label{eq:LIH}
  \later \left(\forall n : \NAT, \everysecond (\BOX \iota . \iterate\,(\NEXT f)\, n )
    =_{\gStream{\NAT}} \iterate\,(\NEXT f^2)\, n\right),
\end{equation}
then by L\"ob induction we can derive
\begin{align*}
  \everysecond\, (\BOX \iota.&\iterate\,(\NEXT f)\,n) \\
  &= n \consin \NEXT\,(\everysecond\,(\limtail\,(\limtail\,(\BOX\iota.
    \iterate\,(\NEXT f)\,n)))) \\
  &= n \consin \NEXT\, (\everysecond\, (\BOX \iota. \iterate\, 
    (\NEXT f)\, (f \, (f\, n))))
    \tag*{(by \ref{eq:iterate-tail})} \\
  &= n \consin \NEXT\, (\iterate\,(\NEXT f^2)\,(f\,(f\, n)))
    \tag*{(by \ref{eq:LIH} and \eqlaternextrule)} \\
  &= \iterate \, (\NEXT f^2)\, n.
\end{align*}
\subsection{Higher-Order Logic Example}
\label{sec:hol-example}

We now prove
\begin{align*}
  &\forall P, Q : (\NN \to \Omega), \forall f : \NN \to \NN,
  (\forall x : \NN, P(x) \implies Q(f(x)))\\
  &\implies\forall xs : \Stream{\NN}, \liftstr{P}(xs) \implies \liftstr{Q}(\map f\,xs).
\end{align*}
This is a simple property of $\map$, but the proof shows how the pieces fit together.
Recall that $\map$ satisfies $\map f\,xs = f\,(\head xs) \consin (\NEXT (\map f) \APP (\tail xs))$.
We prove the property by L\"ob induction. So let $P$ and $Q$ be predicates on $\NN$ and
$f$ a function on $\NN$ that satisfies $\forall x : \NN, P(x) \implies Q(f(x))$. To use
L\"ob induction assume
\begin{align}
  \label{eq:lob-IH-map}
  \later(\forall xs : \Stream{\NN}, \liftstr{P}(xs) \implies \liftstr{Q}(\map f\,xs))
\end{align}
and let $xs$ be a stream satisfying $\liftstr{P}$. Unfolding $\liftstr{P}(xs)$ we get
$P(\head xs)$ and $\lift (\NEXT \liftstr{P} \APP (\tail xs))$ and we need to prove
$Q(\head (\map f\,xs))$ and also $\lift (\NEXT \liftstr{Q} \APP (\tail (\map f\,xs)))$.
The first is easy since $Q(\head (\map f\,xs)) = Q(f\,(\head xs))$. For the second we have
$\tail (\map f\,xs) = \NEXT (\map f) \APP (\tail xs)$. Since $\Stream{\NN}$ is a total and
inhabited type there is a stream $xs'$ such that $\NEXT xs' = \tail xs$. This gives
$\tail (\map f\,xs) = \NEXT (\map f xs')$ and so our desired result reduces to
$\lift (\NEXT (\liftstr{Q}(\map f\, xs')))$ and $\lift (\NEXT \liftstr{P} \APP (\tail
xs))$ is equivalent to $\lift (\NEXT (\liftstr{P}(xs')))$. Now $\lift \comp \NEXT =
\later$ and so what we have to prove is $\later(\liftstr{Q}(\map f\,xs'))$ from
$\later (\liftstr{P}(xs'))$, which follows directly from the induction hypothesis
\eqref{eq:lob-IH-map}.


%% file: appx_sums.tex
\section{Sums}\label{app:sums}

This appendix extends Secs.~\ref{sec:calculus}, \ref{sec:denot} and~\ref{sec:logic}
to add sum types to the $\lambdanext$-calculus. and to logic $\logiclambdanext$.

Binary sums in Atkey and McBride~\cite{Atkey:Productive} come with the type
isomorphism $\blacksquare A+\blacksquare B\cong \blacksquare(A+B)$, but there are
not in general terms witnessing this isomorphism. Likewise if binary sums are added to
our calculus as obvious we may define the term
\[
  \lambda x.\BOX \iota.\CASE x\OF x_1.\IN_1\UNBOX x_1;x_2.\IN_2\UNBOX x_2:
  \blacksquare A+\blacksquare B\to \blacksquare(A+B)
\]
but no inverse is definable in general. We believe such a map may be useful when
working with guarded recursive types involving sum, such as the type of potentially
infinite lists, and in any case the isomorphism is valid in the topos of trees and so it is
harmless for us to reflect this in our calculus. We do this via a new term-former
$\BOXSUM$ allowing us to define
\[
  \lambda x.\BOXSUM \iota.\UNBOX x:\blacksquare(A+B)\to\blacksquare A+\blacksquare B
\]
This construct may be omitted without effecting the results of this section.

\begin{definition}[ref. Defs.~\ref{def:terms},\ref{def:redrule},\ref{def:value},%
  \ref{def:eval_ctx},\ref{def:types},\ref{def:typing}]
$\lambdanext$-\emph{terms} are given by the grammar
\[
  \begin{array}{rcl}
    t & \bnfeq & \cdots ~|~ \ABORT t ~|~ \IN_d t ~|~ \CASE t \OF x_1. t; x_2. t ~|~
      \BOXSUM \sigma.t
  \end{array}
\]
where $d\in\{1,2\}$, and $x_1,x_2$ are variables. We abbreviate terms with
$\BOXSUM$ as for $\PREV$ and $\BOX$.

The \emph{reduction rules} on closed $\lambdanext$-terms with sums are
\[
  \begin{array}{rcll}
    \CASE \IN_d t \OF x_1 . t_1; x_2 . t_2 & \red & t_d[t/x_d]
      & \quad\mbox{\emph{($d\in\{1,2\}$)}} \\
    \BOXSUM[\vec{x} \explsubst \vec{t}].t & \red & \BOXSUM t[\vec{t}/\vec{x}] &
      \quad\mbox{\emph{($\vec{x}$ non-empty)}} \\
    \BOXSUM\IN_i t & \red & \IN_i\BOX t
  \end{array} 
\]

  \emph{Values} are terms of the form
  \[
    \cdots  ~|~ \IN_1 t ~|~ \IN_2 t
  \]

 \emph{Evaluation contexts} are defined by the grammar
  \[
    \begin{array}{rcl}
      E &\bnfeq& \cdots ~|~ \ABORT E ~|~
      \CASE E\OF x_1.t_1;x_2.t_2 ~|~ \BOXSUM E
    \end{array} 
  \]

  $\lambdanext$-\emph{types} for sums are defined inductively by the rules of
  Fig.~\ref{fig:types_sums}, and the new \emph{typing judgments} are given in
  Fig.~\ref{fig:typing-sums}, where $d\in\{1,2\}$.
\end{definition}

  \begin{figure}
  \begin{mathpar}
    \inferrule*{ }{\nabla \vdash \EMPTY}
    \and
    \inferrule*{%
      \nabla \vdash A_1 \\
      \nabla \vdash A_2}{%
      \nabla\vdash A_1+A_2}
  \end{mathpar}
  \caption{Type formation for sums in the $\lambdanext$-calculus}
  \label{fig:types_sums}
  \end{figure}

  \begin{figure}
  \begin{mathpar}
    \inferrule*{\Gamma \vdash t:\EMPTY}{%
      \Gamma \vdash \ABORT t:A}
    \and
    \inferrule*{\Gamma \vdash t: A_d}{%
      \Gamma \vdash \IN_d t : A_1 + A_2}
    \and
    \inferrule*{\Gamma \vdash t: A_1+A_2 \\
      \Gamma,x_1:A_1\vdash t_1:A \\
      \Gamma,x_2:A_2\vdash t_2:A}{%
      \Gamma \vdash \CASE t \OF x_1 . t_1; x_2 . t_2:A}
    \and
    \inferrule*[right={$A_1,\ldots,A_n\,\mathsf{constant}$}]{%
      x_1:A_1,\ldots,x_n:A_n \vdash t:B_1+B_2 \\
      \Gamma\vdash t_1:A_1 \\
      \cdots \\
      \Gamma\vdash t_n:A_n }{%
      \Gamma \vdash \BOXSUM[x_1\explsubst t_1,\ldots,x_n\explsubst t_n].t:
        \blacksquare B_1+\blacksquare B_2}
  \end{mathpar}
  \caption{Typing rules for sums in the $\lambdanext$-calculus}
  \label{fig:typing-sums}
  \end{figure}

We now consider denotational semantics. Note that the initial object of $\trees$ is
$\Delta\emptyset$ (ref. Def.~\ref{def:functors}), while binary coproducts in $\trees$ are
defined pointwise. By naturality it holds that for any arrow $f:X\to Y+Z$ and $x\in X$,
$f_i(x)$ must be an element of the same side of the sum for all $i$.

\begin{definition}[ref. Defs.~\ref{def:types_denote},\ref{def:terms_denote}]
\begin{itemize}
\item
  $\den{\EMPTY}$ is the constant functor $\Delta\emptyset$;
\item
  $\den{A_1+ A_2}(\vec{W})=\den{A_1}(\vec{W})
  +\den{A_2}(\vec{W})$ and likewise for $\trees$-arrows.
\end{itemize}

Term-formers for sums are intepreted via $\trees$-coproducts, with $\ABORT$, $\IN_d$
and $\CASE$ defined as usual, and $\BOXSUM$ defined as follows.
\begin{itemize}
\item
  Let $\den{t}_j(\den{t_1}_i(\gamma),\ldots,\den{t_n}_i(\gamma))$ (which is
  well-defined
  by Lem.~\ref{lem:const_types}) be $[a_j,d]$ as $j$ ranges, recalling that
  $d\in\{1,2\}$ is the same for all $i$. Define $a:1\to\den{A_d}$ to have $j$'th element
  $a_j$. Then $\den{\BOXSUM[\vec{x}\explsubst\vec{t}].t}_i(\gamma)\defeq[a,d]$.
\end{itemize}
\end{definition}

We now proceed to the sum cases of our proofs.

\begin{proof}[$\BOXSUM[\vec{y}\explsubst\vec{u}{]}.t$ case of Lem.~\ref{lem:subst}]
By induction we have $\den{u_k[\vec{t}/\vec{x}]}_i(\gamma)=\den{u_k}_i(\den{t_1}_i
(\gamma),\ldots)$. Hence $\den{t}_j(\den{u_1[\vec{t}/\vec{x}]}_i(\gamma),\ldots)=
\den{t}_j(\den{u_1}_i(\den{t_1}_i(\gamma),\ldots),\ldots)$ as required.
\end{proof}

\begin{proof}[$\BOXSUM$ cases of Soundness Thm.~\ref{lem:soundness}]
Because each $\den{A_k}$ is a constant object (Lem.~\ref{lem:const_types}),
$\den{t_k}_i=\den{t_k}_j$ for all $i,j$. Hence $\den{\BOXSUM[\vec{x}\explsubst
\vec{t}].t}_i$ is defined via components $\den{t}_j(\den{t_1}_j,\ldots)$ and
$\den{\BOXSUM t[\vec{t}/\vec{x}]}$ is defined via components $\den{t[\vec{t}/
\vec{x}]}_j$. These are equal by Lem~\ref{lem:subst}.

$\den{\BOXSUM\IN_d t}_i$ is the $d$'th injection into the function with $j$'th
component $\den{t}_j$, and likewise for $\den{\IN_d\BOX t}_i$.
\end{proof}

\begin{definition}[ref. Def.~\ref{Def:log_rel}]
\begin{itemize}
\item $[a,d]\lrel{i}{A_1+A_2} t$ iff $t\redrt \IN_d u$ for $d=1$ or $2$, and $a
  \lrel{i}{A_d}u$.
\end{itemize}
Note that $\lrel{i}{\EMPTY}$ is (necessarily) everywhere empty.
\end{definition}

\begin{proof}[for Lems.~\ref{lem:lrel_and_red} and ~\ref{lem:res_and_rel}]
For $\EMPTY$ cases the premise fails so the the lemmas are vacuous. $+$ cases
follow as for $\times$.
\end{proof}

\begin{proof}[ref. Fundamental Lemma~\ref{lem:ad}]
$\ABORT$: The induction hypothesis states that $\den{t}_k(\vec{a})\lrel{k}{\EMPTY}
t[\vec{t}/\vec{x}\,]$, but this is not possible, so the theorem holds vacuously.

$\IN_d t$ case follows by easy induction.

$\CASE t\OF y_1.u_1;y_2.u_2$: If $\den{t}_i(\vec{a})\lrel{i}{A_1+A_2}t[\vec{t}/
\vec{x}]$ then $t[\vec{t}/\vec{x}\,]\redrt\IN_d u$ for some $d\in\{1,2\}$, with
$\den{t}_i(\vec{a})=[a,d]$ and $a\lrel{i}{A_d}u$. Then $\den{u_d}_i(\vec{a},a)
\lrel{k}{A}u_d[\vec{t}/\vec{x},u/y_d]$. Now $(\CASE t\OF y_1.u_1;y_2.u_2)[\vec{t}/
\vec{x}]\redrt\CASE\IN_d u\OF y_1.(u_1[\vec{t}/\vec{x}]);y_2.(u_2[\vec{t}/\vec{x}])$,
which reduces to $u_d[\vec{t}/\vec{x},u/y_i]$, and Lem.~\ref{lem:lrel_and_red}
completes.

$\BOXSUM[\vec{y}\explsubst\vec{u}].t$: $\den{u_k}_i(\vec{a})\lrel{i}{A_k}u_k[\vec{t}/
\vec{x}]$ by induction, so $\den{u_k}_i(\vec{a})\lrel{j}{A_k}u_k[\vec{t}/\vec{x}]$
for any $j$ by Lem.~\ref{lem:constrel}. By induction $\den{t}_j(\den{u_1}_k
(\vec{a}),\ldots)\lrel{j}{B_1+B_2}t[u_1[\vec{t}/\vec{x}\,]/y_1,\ldots]$. If
$\den{t}_j(\den{u_1}_k(\vec{a}),\ldots)$ is some $[b_j,d]$ we have $t[u_1[\vec{t}/
\vec{x}]/y_1,\ldots]\redrt\IN_d s$ with $b_j\lrel{j}{B_d}s$. Now $(\BOXSUM[\vec{y}
\explsubst\vec{u}].t)[\vec{t}/\vec{x}\,]\red \BOXSUM t[u_1[\vec{t}/\vec{x}\,]/y_1,\ldots]
\redrt \BOXSUM\IN_d s$, which finally reduces to $\IN_d\BOX s$, which yields the
result.
\end{proof}

The logic $\logiclambdanext$ may be extended to sums via the usual $\beta\eta$-laws
and commuting conversions for binary sums and the equational version of the
$\BOXSUM$ rule (ref. Fig.~\ref{fig:additional-equations}):
\[
    \inferrule*{%
      \Gamma_{\blacksquare} \vdash t:B_d \\
      \Gamma\vdash \vec{t} : \Gamma_{\blacksquare}}{%
      \Gamma \vdash \BOXSUM[\vec{x}\explsubst\vec{t}].(\IN_d t) = \IN_d(\BOX [\vec{x}\explsubst\vec{t}] . t )}
\]


%% file: appx_bde_proof.tex
\section{Proof of Definability of Solutions of Behavioural Differential Equations in $\lambdanext$}
\label{sec:proof-bde}

An equivalent presentation of the topos of trees is as sheaves over $\omega$ (with
Alexandrov topology) $\Sh{\omega}$. In this section it is more convenient to work with
sheaves than with presheaves because the global sections functor $\Gamma$%
\footnote{This standard notation for this functor should not to be confused with our notation
for typing contexts.}
in the sequence of adjoints
\begin{align*}
  \Pi_1 \dashv \Delta \dashv \Gamma
\end{align*}
where
\begin{align*}
  \begin{split}
    \Pi_1 &: \Sl \to \sets\\
    \Pi_1(X) &= X(1)
  \end{split}
  \qquad
  \begin{split}
    \Delta &: \sets \to \Sl\\
    \Delta(a)(\alpha) &=
    \begin{cases}
      1 & \text{if } \alpha = 0\\
      a & \text{otherwise}
    \end{cases}
  \end{split}
  \qquad
  \begin{split}
    \Gamma &: \Sl \to \sets\\
    \Gamma(X) &= X(\omega)
  \end{split}
\end{align*}
is just evaluation at $\omega$, i.e. the limit is already present. This simplifies
notation. Another advantage is that $\LATER : \Sl \to \Sl$ is given as
\begin{align*}
  (\LATER X)(\nu+1) &= X(\nu)\\
  (\LATER X)(\alpha) &= X(\alpha)
\end{align*}
where $\alpha$ is a limit ordinal (either $0$ or $\omega$) which means that $\LATER
X(\omega) = X(\omega)$ and as a consequence, $\nxt_{\omega} = \id{X(\omega)}$ and
$\Gamma(\LATER X) = \Gamma(X)$ for any $X \in \Sl$ and so $\blacksquare (\LATER X) =
\blacksquare X$ for any $X$ so we don't have to deal with mediating isomorphisms.

First we have a simple statement, but useful later, since it gives us a precise goal to
prove later when considering the interpretation.
\begin{lemma}
  \label{lem:fixed-point-mapped-to-fp}
  Let $X, Y$ be objects of $\Sl$. Let $F : \LATER\left(Y^X\right) \to Y^X$ be a morphism
  in $\Sl$ and $\underline F$ a \emph{function} in $\sets$ from $Y(\omega)^{X(\omega)}$ to
  $Y(\omega)^{X(\omega)}$. Suppose that the diagram
  \[
  \begin{largediagram}
    \Gamma\left(\LATER\left(Y^X\right)\right) \ar{r}[description]{\Gamma(F)} \ar{d}[description]{\lim} & \Gamma(Y^X) \ar{d}[description]{\lim}\\
    Y(\omega)^{X(\omega)} \ar{r}[description]{\underline F} & Y(\omega)^{X(\omega)}
  \end{largediagram}
  \]
  where $\lim\left(\{g_\nu\}_{\nu=0}^{\omega}\right) = g_{\omega}$ commutes.  By Banach's
  fixed point theorem $F$ has a unique fixed point, say $u : 1 \to Y^X$.

  Then $\lim(\Gamma(u)(\ast)) = \lim(\Gamma(\nxt \comp u)(\ast)) =
  \Gamma(\nxt \comp u)(\ast)_{\omega} = u_{\omega}(\ast)_{\omega}$ is a fixed point of $\underline F$.
\end{lemma}
\begin{proof}
  The proof is trivial.
  \begin{align*}
    \underline F \left(\lim(\Gamma(u)(\ast))\right) &=
    \lim(\Gamma(F)(\Gamma(\nxt \comp u)(\ast)))\\
    &= \lim(\Gamma(F \comp \nxt \comp u)(\ast)) = 
    \lim(\Gamma(u)(\ast)).
  \end{align*}
\end{proof}

Note that $\lim$ is not an isomorphism. There are (in general) many more functions from
$X(\omega)$ to $Y(\omega)$ than those that arise from natural transformations. The ones
that arise from natural transformations are the \emph{non-expansive} ones.

\subsection{Behavioural Differential Equations}
\label{sec:rutt-behav-diff}

Let $\Sigma_A$ be a signature of function symbols with two types, $A$ and $\Stream{A}$.
Suppose we wish to define a new $k$-ary operation given the signature $\Sigma_A$. We need
to provide two terms $h_f$ and $t_f$ (standing for \emph{head} and \emph{tail}).  $h_f$
has to be a term using function symbols in signature $\Sigma_A$ and have type
\begin{align*}
  \hastype{x_1 : A, x_2 : A, \cdots, x_k : A}{h_f}{A}
\end{align*}
and $t_f$ has to be a term in the signature extended with a new function symbol $f$ of
type $\left(\Stream{A}\right)^k \to \Stream{A}$ and have type
\begin{align*}
  \hastype{x_1 : A, \cdots, x_k : A, y_1 : \Stream{A}, \cdots, y_k : \Stream{A}, z_1 :
    \Stream{A}, \cdots, z_k : \Stream{A}}{t_f}{\Stream{A}}
\end{align*}
In the second term the variables $x$ (intuitively) denote the head elements of the
streams, the variables $y$ denote the streams, and the variables $z$ denote the tails of
the streams.

We now define two interpretations of $h_f$ and $t_f$. First in the topos of trees and then
in $\sets$.

We choose a set $a \in \sets$ and define $\denS{A} = \Delta(a)$ and $\denS{\Stream{A}} =
\mu X . \Delta(a) \times \LATER(X)$.  To each function symbol $g \in \Sigma$ of type
$\tau_1, \ldots, \tau_n \to \tau_{n+1}$ we assign a morphism
\begin{align*}
  \denS{g} : \denS{\tau_1} \times \denS{\tau_2} \times \cdots \times \denS{\tau_n} \to
  \denS{\tau_{n+1}}.
\end{align*}
Then we define the interpretation of $h_f$ by induction as a morphism of type
$\denS{A}^k \to \denS{A}$ by
\begin{align*}
  \denS{x_i} &= \pi_i\\
  \denS{g(t_1, t_2, \ldots, t_n)} &= \denS{g} \comp \left\langle \denS{t_1}, \denS{t_2},
    \cdots, \denS{t_n}\right\rangle.
\end{align*}
For $t_f$ we interpret the types and function symbols in $\Sigma_A$ in the same way. But
recall that $t_f$ also contains a function symbol $f$. So the denotation of $t_f$ will be
a morphism with the following type
\begin{align*}
  \denS{t_f} : \LATER\left(\denS{\Stream{A}}^{\denS{\Stream{A}}^k}\right) \times \denS{A}^k
  \times \denS{\Stream{A}}^k \times \left(\LATER\left(\denS{\Stream{A}}\right)\right)^k \to
  \LATER(\denS{\Stream{A}})
\end{align*}
and is defined as follows
\begin{align*}
  \denS{x_i} &= \nxt \comp \iota \comp \pi_{x_i}\\
  \denS{y_i} &= \nxt \comp \pi_{y_i}\\
  \denS{z_i} &= \pi_{z_i}\\
  \denS{g(t_1, t_2, \ldots, t_n)} &= \LATER(\denS{g}) \comp \mathbf{can} \comp
  \left\langle \denS{t_1}, \denS{t_2},\cdots, \denS{t_n}\right\rangle &\text{if } g \neq f\\
  \denS{f(t_1, t_2, \ldots, t_k)} &=
  \mathbf{eval} \comp \left\langle J \comp \pi_f,
    \mathbf{can} \comp \left\langle \denS{t_1}, \denS{t_2}, \cdots, \denS{t_k}\right\rangle\right\rangle
\end{align*}
where $\mathbf{can}$ is the canonical isomorphism witnessing that $\LATER$ preserves
products, $\mathbf{eval}$ is the evaluation map and $\iota$ is the suitably encoded
morphism that when given $a$ constructs the stream with head $a$ and tail all zeros. This
exists and is easy to construct.

Next we define the denotation of $h_f$ and $t_f$ in $\sets$. We set $\denSet{A} = a$ and
$\denSet{\Stream{A}} = \denS{\Stream{A}}(\omega)$.  For each function symbol in $\Sigma_A$
we define $\denSet{g} = \Gamma{\denS{g}} = \left(\denS{g}\right)_{\omega}$.

We then define $\denSet{h_f}$ as a function
\[\denSet{A}^k \to \denSet{A}\]
exactly the same as we defined $\denS{h_f}$.
\begin{align*}
  \denSet{x_i} &= \pi_i\\
  \denSet{g(t_1, t_2, \ldots, t_n)} &= \denSet{g} \comp \left\langle \denSet{t_1}, \denSet{t_2},
    \cdots, \denSet{t_n}\right\rangle.
\end{align*}

The denotation of $t_f$ is somewhat different in the way that we do not guard the tail
and the function being defined with a $\LATER$. We define
\begin{align*}
  \denSet{t_f} : \denSet{\Stream{A}}^{\denSet{\Stream{A}}^k} \times \denSet{A}^k \times
  \denSet{\Stream{A}}^k \times \left(\denSet{\Stream{A}}\right)^k \to \denSet{\Stream{A}}
\end{align*}
as follows
\begin{align*}
  \denSet{x_i} &= \iota \comp \pi_{x_i}\\
  \denSet{y_i} &= \pi_{y_i}\\
  \denSet{z_i} &= \pi_{z_i}\\
  \denSet{g(t_1, t_2, \ldots, t_n)} &= \denSet{g} \comp
  \left\langle \denSet{t_1}, \denSet{t_2},\cdots, \denSet{t_n}\right\rangle &\text{if } g \neq f\\
  \denSet{f(t_1, t_2, \ldots, t_k)} &=
  \mathbf{eval} \comp \left\langle \pi_f,
    \left\langle \denSet{t_1}, \denSet{t_2}, \cdots, \denSet{t_k}\right\rangle\right\rangle
\end{align*}
where $\iota$ is again the same operation, this time on actual streams in $\sets$.

We then define
\[
\underline F : \denSet{\Stream{A}}^{\denSet{\Stream{A}}^k} \to
\denSet{\Stream{A}}^{\denSet{\Stream{A}}^k}
\]
as
\begin{align*}
  \underline F(\phi)\left(\vec{\sigma}\right) = \Gamma\left(\mathbf{fold}\right)\left(\left(
  \denSet{h_f}\left(\mathbf{hd}(\vec{\sigma})\right), \denSet{t_f}\left(\phi,
    \mathbf{hd}(\vec{\sigma}), \vec{\sigma}, \mathbf{tl}(\vec{\sigma})\right)\right)\right)
\end{align*}
where $\mathbf{hd}$ and $\mathbf{tl}$ are head and tail functions (extended in the obvious
way to tuples). Here $\mathbf{\mathbf{fold}}$ is the isomorphism witnessing that guarded
streams are indeed the fixed point of the defining functor.

Similarly we define
\[
F : \LATER\left(\denS{\Stream{A}}^{\denS{\Stream{A}}^k}\right) \to
\denS{\Stream{A}}^{\denS{\Stream{A}}^k}
\]
as the exponential transpose $\Lambda$ of
\begin{align*}
  F' = \mathbf{fold} \comp \left\langle \den{h_f} \comp \vec{hd} \comp \pi_2, \denS{t_f} \comp
  \left(\id{\LATER\left(\denS{\Stream{A}}^{\denS{\Stream{A}}^k}\right)} \times
    \left\langle \vec{\mathbf{hd}}, \id{\denS{\Stream{A}}^k}, \vec{\mathbf{tail}}\right\rangle\right)\right\rangle
\end{align*}


\begin{proposition}
  \label{prop:stream-diagram-commutes}
  For the above defined $F$ and $\underline F$ we have
  \begin{align*}
    \lim \comp \Gamma(F) = \underline F \comp \lim
  \end{align*}
\end{proposition}
\begin{proof}
  Let $\phi \in
  \Gamma\left(\LATER\left(\denS{\Stream{A}}^{\denS{\Stream{A}}^k}\right)\right) =
  \Gamma\left(\denS{\Stream{A}}^{\denS{\Stream{A}}^k}\right)$.
  We have
  \begin{align*}
    \lim(\Gamma(F)(\phi)) = \lim\left(F_{\omega}(\phi)\right) = F_{\omega}(\phi)_\omega
  \end{align*}
  and
  \begin{align*}
    \underline F (\lim(\phi)) = \underline F \left(\phi_{\omega}\right)
  \end{align*}
  Now both of these are elements of $\denSet{\Stream{A}}^{\denSet{\Stream{A}}^k}$, meaning
  genuine functions in $\sets$, so to show they are equal we use elements. Let
  $\vec{\sigma} \in \denSet{\Stream{A}}^k$.

  We are then required to show
  \begin{align*}
    \underline F \left(\phi_{\omega}\right)(\vec{\sigma}) = F_{\omega}(\phi)_\omega(\vec{\sigma})
  \end{align*}
  Recall that $F = \Lambda(F')$ (exponential transpose) so 
  $F_{\omega}(\phi)_\omega(\vec{\sigma}) = F'_{\omega}(\phi, \vec{\sigma})$. Now recall
  that composition in $\Sl$ is just composition of functions at each stage and products in
  $\Sl$ are defined pointwise and that $\nxt_{\omega}$ is the identity function.
  
  Moreover, the morphism $\mathbf{hd}$ gets mapped by $\Gamma$ to $\mathbf{hd}$ in
  $\sets$ and the same holds for $\mathbf{tl}$. For the latter it is important that
  $\Gamma(\LATER(X)) = \Gamma(X)$ for any $X$.

  We thus get
  \begin{align*}
    F'_{\omega}(\phi, \vec{\sigma}) = \mathbf{fold}_{\omega}\left((\denS{h_f})_{\omega}
    \left(\mathbf{hd}(\vec{\sigma})\right), \left(\denS{t_f}\right)_{\omega}\left(\phi,
      \mathbf{hd}(\vec{\sigma}), \vec{\sigma}, \mathbf{tl}(\vec{\sigma})\right)\right)
  \end{align*}

  And for $\underline F \left(\phi_{\omega}\right)(\vec{\sigma})$ we have
  \begin{align*}
    \underline F \left(\phi_{\omega}\right)(\vec{\sigma}) =
    \mathbf{fold}_{\omega}\left(\denSet{h_f}\left(\mathbf{hd}\left(\vec{\sigma}\right)\right),
      \left(\denSet{t_f}\right)\left(\phi_{\omega}, \mathbf{hd}(\vec{\sigma}), \vec{\sigma},
        \mathbf{tl}(\vec{\sigma})\right)\right)
  \end{align*}
  It is now easy to see that these two are equal. The proof is by induction on the
  structure of $h_f$ and $t_f$. The variable cases are trivial, but crucially use the fact
  that $\nxt_{\omega}$ is the identity. The cases for function symbols in $\Sigma_A$ are
  trivial since their denotations in $\sets$ are defined to be the correct ones. The case
  for $f$ goes through similarly since application at $\omega$ only uses $\phi$ at
  $\omega$.
\end{proof}

\begin{theorem}
  \label{thm:diff-equations-fp}
  Let $(\Sigma_1, \Sigma_2)$ be a signature and $\Il$ its interpretation.  Let $(h_f,
  t_f)$ be a behavioural differential equation defining a $k$-ary function $f$ using
  function symbols in $\Sigma$. The right-hand sides of $h_f$ and $t_f$ define a term
  $\guarded{\Phi_f}$ of type
  \begin{align*}
    \guarded{\Phi_f} :
    \LATER (\underbrace{\gStream{\NAT} \to \gStream{\NAT} \to \cdots \gStream{\NAT}}_{k+1})
    \to (\underbrace{\gStream{\NAT} \to \gStream{\NAT} \to \cdots \gStream{\NAT}}_{k+1}).
  \end{align*}
  and a term $\Phi_f$ of type
  \begin{align*}
    \Phi_f : (\underbrace{\Stream{\NAT} \to \Stream{\NAT} \to \cdots \Stream{\NAT}}_{k+1})
    \to (\underbrace{\Stream{\NAT} \to \Stream{\NAT} \to \cdots \Stream{\NAT}}_{k+1}).
  \end{align*}
  by using $\mathcal{L}_{a^g_j}\left(\Il(g_j)\right)$ for interpretations of function
  symbols $g_j$.
  
  Let $\guarded{f} = \Theta\guarded{\Phi_f}$ be the fixed point of $\guarded{\Phi_f}$.
  Then $f = \mathcal{L}_k(\BOX \guarded f)$ is a fixed point of $\Phi_f$ which in turn
  implies that it satisfies equations $h_f$ and $t_f$.
\end{theorem}
\begin{proof}
  Use Prop.~\ref{prop:stream-diagram-commutes} together with
  Lemma~\ref{lem:fixed-point-mapped-to-fp} together with the observation that $\sets$ is a
  full subcategory of $\Sl$ with $\Delta$ being the inclusion.
  
  We also use the fact that for a closed term $u : A \to B$ (which is interpreted as a
  morphism from $1$ to $B^A$) the denotation of $\mathcal{L}(u)$ at stage
  $\nu$ and argument $\ast$ is $\lim(\Gamma(u)(\ast))$.
\end{proof}
\subsection{Discussion}
\label{sec:rutten-discussion}

What we have shown is that for each behavioural differential equation that defines a
function on streams and \emph{can be specified as a standalone function depending only on
  previously defined functions}, i.e. it is not defined mutually with some other function,
there is a fixed point. It is straightforward to extend to mutually recursive definitions
by defining a product of functions in the same way as we did for a single function, but
notationally this gets quite heavy.

More importantly, suppose we start by defining an operation $f$ on
streams first, and the only function symbols in $\Sigma_A$ operate on $A$, i.e. all have
type $A^k \to A$ for some $k$. Assume that these function symbols are given denotations in
$\Sl$ as $\Delta(g)$ for some function $g$ in $\sets$. Then the denotation in $\sets$ is
just $g$.

The fixed point $f$ \emph{in $\Sl$} is then a morphism from $1$ to the suitable
exponential. Let $\overline f$ be the uncurrying of $f$.  Then $\lim(\Gamma(f)(\ast)) =
\Gamma(\overline f)$.

Thus if we continue defining new functions which use $f$, we then choose $\overline f$ as
the denotation of the function symbol $f$. The property $\lim(\Gamma(u)(\ast)) =
\Gamma(\overline f)$ then says that the $f$ that is used in the definition is the $f$ that
was defined previously.


%% file: appx_total_inhab.tex
\section{About Total and Inhabited Types}
\label{sec:about-totality-types}

An object in $\trees$ is \emph{total and inhabited} if all components are non-empty and
all restriction functions are surjective. We have the following easy proposition.
\begin{proposition}
  \label{prop:functor-restr-fp-total}
  Let $P : \Sl \to \Sl$ be a functor such that if $X$ is a total and inhabited object, so
  is $P(X)$, i.e. $P$ restricts to the full subcategory of total and inhabited objects.

  If $P$ is locally contractive then its fixed point is total and inhabited.
\end{proposition}
\begin{proof}
  We use the equivalence between the full subcategory $ti\Sl$ of $\Sl$ of
  \emph{total and inhabited}
  objects and the category of complete bisected \emph{non-empty} ultrametric
  spaces $\Ml$. We know that the category $\Ml$ is an $M$-category%
\footnote{Birkedal, L., St{\o}vring, K., Thamsborg, J.: The category-theoretic solution of recursive metric-space equations. Theor. Comput. Sci. 411(47), 4102--4122 (2010)}
  and thus so is
  $ti\Sl$. It is easy to see that locally contractive functors in $\Sl$ are locally
  contractive in the $M$-category sense. Hence if $P$ is locally contractive and
  restricts to $ti\Sl$ its fixed point is in $ti\Sl$.
\end{proof}

\begin{corollary}
  \label{cor:polynomial-fp-total-inhabited}
  Let $P$ be a non-zero polynomial functor whose coefficients and exponents are total and
  inhabited. The functor $P \comp \LATER$ is locally contractive and its unique fixed
  point is total and inhabited.
\end{corollary}
\begin{proof}
  Products and non-empty coproducts of total and inhabited objects are total and
  inhabited. Similarly, if $X$ and $Y$ are total and inhabited, so is $X^Y$. So any
  non-zero polynomial functor $P$ whose coefficients are all total and inhabited restricts
  to $ti\Sl$. The functor $\LATER$ restricts to $ti\Sl$ as well (but note that it
  \emph{does not} restrict to the subcategory of total objects $t\Sl$).  Polynomial
  functors on $\Sl$ are also strong and so the functor $P \comp \LATER$ is locally
  contractive. Hence by Prop.~\ref{prop:functor-restr-fp-total} its unique fixed
  point is a total and inhabited object.
\end{proof}
In particular guarded streams of any total inhabited type themselves form a total and
inhabited type.